\documentclass[12pt]{article}
\usepackage{amsmath,amssymb,amsthm}
\usepackage[margin=1in]{geometry}
\usepackage[colorlinks,citecolor=blue]{hyperref}
\usepackage{natbib} 
\usepackage{enumitem} 
\usepackage{adjustbox}
\usepackage[flushleft]{threeparttable}
\usepackage{multirow, rotating,setspace}
\usepackage{threeparttable}
\usepackage{booktabs}
\usepackage{graphicx}
\usepackage{subcaption}

\allowdisplaybreaks

\renewcommand{\qed}{\rule{2mm}{2mm}}

\DeclareMathOperator*{\cov}{Cov}

\DeclareMathOperator*{\E}{E}
\DeclareMathOperator*{\Var}{Var}

\newtheorem{theorem}{Theorem}[section]
\newtheorem{lemma}{Lemma}[section]

\newtheorem{proposition}{Proposition}[section]

\theoremstyle{plain} 
\newtheorem{definition}{Definition}[section]
\newtheorem{procedure}{Procedure}[section]

\theoremstyle{definition} 
\newtheorem{example}{Example}[section]
\newtheorem{remark}{Remark}[section]
\newtheorem{assumption}{Assumption}[section]
\newtheorem{hypothesis}{Hypothesis}[section]

\AtEndEnvironment{remark}{~\qed}
\AtEndEnvironment{example}{}

\newcommand{\W}{\mathbf{W}}
\newcommand{\ww}{\mathbf{w}}

\newcommand{\Wobs}{\W^{obs}}

\newcommand{\Hbuyer}{H_0^{\text{buyer}}}
\newcommand{\Hseller}{H_0^{\text{seller}}}
\newcommand{\Htotal}{H_0^{\text{total}}}

\usepackage{xcolor, color, ulem}

\begin{document}
\def\spacingset#1{\renewcommand{\baselinestretch}%
{#1}\small\normalsize} \spacingset{1}

\author{
Jizhou Liu \\
PHBS Business School\\
Peking University\\
\url{jizhou.liu@phbs.pku.edu.cn}
\and 
Azeem M.\ Shaikh \\
Department of Economics \\
University of Chicago\\
\url{amshaikh@uchicago.edu}
\and
Panos Toulis \\
Booth School of Business\\
University of Chicago\\
\url{panos.toulis@chicagobooth.edu}
}

\bigskip

\title{Randomization Inference in Two-Sided Market Experiments \thanks{The second author acknowledges support from the National Science Foundation through grant SES-2419008.}}

\maketitle

\onehalfspacing

\begin{abstract}
Randomized experiments are increasingly employed in two-sided markets, such as buyer--seller platforms, to evaluate the effects of marketplace interventions. These experiments must reflect the underlying two-sided market structure in their design and can therefore be challenging to analyze. In this paper, we develop a randomization inference framework for outcomes from two-sided experiments, with a focus on testing and inference for two-sided spillover effects. Our approach is finite-sample valid under sharp null hypotheses. Regarding weak null hypotheses, we find that the commonly used Neyman-style studentization does not universally ensure asymptotic validity, and we document how it depends on the specific formulation of the null. We then propose a two-way variance estimator for studentization that restores asymptotic validity. We further propose methods to improve testing power by exploiting the two-sided structure of the problem, which we validate empirically. We demonstrate our methods through a series of simulation studies and an applied example from a network experiment in micro-lending.
\end{abstract}


\hypersetup{pageanchor=false}
\thispagestyle{empty} 
\newpage
\hypersetup{pageanchor=true}
\setcounter{page}{1}

\spacingset{1.7} %

\section{Introduction}
Randomized experiments are increasingly used to assess the impact of policies in various online domains, including online marketplaces, and streaming media platforms~\citep{Kohavi2020,Athey2019,thomke2003experimentation}.
These environments generate complex datasets through two-sided interactions between multiple types of actors, such as viewers and content creators or buyers and sellers. 
In this context, policy evaluation is complicated due to \textit{market interference}, where the policy treatment on certain units may affect the outcomes of different units in the market, including those never assigned to treatment~\citep{whymarketplace, toulis2016long}.
Traditional randomized experiments, designed for single-population settings, often fall short in addressing the complexities arising from interference in two-sided markets~\citep{imbens2021,johari2021,Wager2021}. 
 
To overcome these challenges, a novel class of experimental designs, termed \textit{Multiple Randomization Designs}, was developed independently by \cite{imbens2021,bajari2023} and \cite{johari2021}, specifically for marketplace experimentation. For example, in an online marketplace with buyers and sellers, the experimenter might randomly assign half of the buyers and half of the sellers to treatment groups, applying a policy intervention (e.g., free shipping) to transactions where both the buyer and the seller are treated. Transactions where either the buyer or the seller (or both) are untreated would not receive the intervention. This design creates variation in the exposure to treatment across different transactions, which facilitates the analysis of causal effects in the presence of interference.

In this paper, we propose Fisherian-style randomization tests~\citep{Fisher1935Design} tailored specifically for analyzing two-sided market experiments of this kind. 
We formulate sharp null hypotheses on spillover effects that are naturally aligned with the two-sided structure of these experiments. Our approach adopts the conditional randomization testing framework that restricts attention to a set of ``focal units'' associated with each tested hypothesis~\citep{Athey2018, Basse2019, puelz2021}. In this framework, the spillover hypotheses are not sharp in the traditional sense as they cannot be used to impute all potential outcomes. 
However, with proper conditioning, the hypotheses become sharp and are thus amenable to standard Fisherian-style randomization.

We propose two valid randomization testing procedures: a permutation-based procedure, which is straightforward to implement but relies on specific symmetry conditions on the design that may not universally hold; and an alternative procedure that is valid under arbitrary designs but requires sampling from a potentially complex space. For a third null hypothesis on total treatment effects, we identify a novel validity-power trade-off among multiple candidate tests and provide heuristic guidance on improving power.

Furthermore, we extend our framework to test weak null hypotheses on average spillover and total treatment effects by incorporating studentized statistics into our randomization tests. While studentization is a well-established technique for ensuring asymptotic validity under weak nulls \citep{chung2013, diciccio2017robust, ZHAO2021278, Wu2021, toulis2025asymptotic}, we demonstrate that its effectiveness in two-sided experiments depends crucially on how the null hypothesis is formulated and what variance estimators are used for studentization. Specifically, we show that the commonly used Neyman-style studentization is asymptotically valid under a null hypothesis concerning \textit{focal average} effects, but it may fail to control size under a \textit{global average} null due to the additional variance introduced by the two-sided randomization. Then, we propose a modified, two-way variance estimator for studentization that restores asymptotic validity for the global weak null by absorbing the additional sampling variation. Our key technical contribution involves extending the proof techniques of \cite{ZHAO2021278} to account for these two-sided dependencies, a finding we validate through extensive simulation studies.

This paper contributes to the growing literature on randomized experiments in multi-sided online marketplaces. \cite{bajari2023} pioneered the introduction of multiple randomization designs to account for complex spillover effects occurring within marketplaces. \cite{imbens2021} delves deeper into such designs, adopting a Neymanian perspective that emphasizes randomization-based point estimation and proposes conservative variance estimators for statistical inference. By contrast, our work considers the Fisherian perspective that instead focuses on finite-sample exact $p$-values via randomization-based tests.

The rest of the paper is organized as follows. Section \ref{sec:notation} describes the setup and notation. Section \ref{sec:randomization-tests} presents the main results under sharp null hypotheses. Section \ref{sec:weak-null} discusses the randomization tests for the weak null hypotheses. Section \ref{sec:simulation} examines the finite-sample behavior of our methods through simulations. Section \ref{sec:empirical} illustrates the proposed inference methods in an empirical application based on the experiment conducted in \cite{Comola2021}. Section \ref{sec:conclusion} concludes the paper.

\section{Setup and Notation}\label{sec:notation}
We begin by introducing the two-sided randomized design following the two-population buyer-seller example in \cite{imbens2021}. Consider a marketplace comprising \(I\) buyers and \(J\) sellers, indexed by \(i = 1, \ldots, I\) and \(j = 1, \ldots, J\), respectively. 
In the marketplace, buyers and sellers interact, 
producing an outcome of interest, such as the total number of transactions between 
the buyer-seller pair within a given time period. 
We denote this observed outcome between buyer $i$ and seller $j$ as \(Y_{i,j} \in \mathbb{R}\). We denote the set of buyer-seller pairs as $\mathbb{U}= \{(i,j):i=1,\dots,I, j=1,\dots,J\}$.
The  \(I \times J\) matrix \(\mathbf{Y} = [Y_{i,j}]\) captures the outcomes observed in the marketplace. 

To examine the impact of various interventions, researchers conduct randomized experiments at the level of individual buyer-seller pairs, using a binary treatment assignment, \(W_{i,j} \in \{0, 1\}\). These interventions might include incentives such as free shipping or discounts on transaction fees. The matrix $\mathbf{W} = [W_{i,j}]\in\{0,1\}^{I\times J}$ denotes the full treatment assignment. We use 
$Y_{i,j}(\ww)$ to denote the potential outcome for the pair \((i, j)\) under 
treatment assignment $\ww\in\{0,1\}^{I\times J}$, and $\mathbf{Y}(\ww)$ for the full $I\times J$ matrix of potential outcomes under assignment \(\ww\). As usual, the observed outcomes and the potential outcomes are related to treatment assignment by the consistency relationship $\mathbf{Y}=\mathbf{Y}(\W)$.

As highlighted by \cite{imbens2021}, treatments in marketplace experiments need not be uniformly applied across all buyers associated with a particular seller, nor across all sellers interacting with a specific buyer. In this paper, we recast the Multiple Randomization Design introduced by \cite{bajari2023} as an \textit{independent two-sided randomized design}, emphasizing that two independent randomizations occur separately on the two sides (buyers and sellers) of the marketplace. Following \cite{imbens2021}, a buyer-seller pair—which constitutes the unit of analysis—is considered treated only if both the buyer and the seller are individually assigned treatment. However, we deviate slightly from Definition 3.4 in \cite{imbens2021} by relaxing the requirement that the proportion treated on each side must remain fixed. Consequently, in our setting, the full treatment matrix $\mathbf{W}$ can be represented as the outer product of the buyer treatment vector and the seller treatment vector, formalized in the following definition.

\begin{definition}[Independent Two-Sided Randomized Design]
\label{def:multiple-randomization}
Let \(w^{B} = (w_1^B, \dots, w_I^{B})^{\top} \in \mathbb{W}^B \subseteq \{0,1\}^{I}\) represent a binary treatment assignment for buyers with probability distribution \(p^B\); and \(w^{S} = (w_1^S, \dots, w_J^S)^{\top} \in \mathbb{W}^S \subseteq \{0,1\}^{J}\) denote the treatment assignment for sellers with distribution \(p^S\). An independent two-sided randomized design is then defined by the induced probability distribution of \(\W = w^B (w^S)^{\top} \subseteq \{0,1\}^{I\times J} = \mathbb{W}\), where the buyer and seller assignments are independent.
\end{definition}

\begin{example}
    In equation (\ref{eqn:w}), we present three examples of treatment assignment matrices \(\W_1\), \(\W_2\), and \(\W_3\), corresponding to buyer-side, seller-side, and two-sided designs, respectively. Treatment assignments for a single population, \(\W_1\) and \(\W_2\), can be equivalently represented by the assignment vectors \(w^B = (0, 1, 0, 1)\) and \(w^S = (1, 0, 1, 0, 0)\), as defined in Definition \ref{def:multiple-randomization}. The assignment matrix for a two-sided randomized design, \(\W_3\), is then defined as the element-wise product (Hadamard product) of \(\W_1\) and \(\W_2\), \(\W_3 = \W_1 \odot \W_2\), or as the outer product of the two assignment vectors, \(\W_3 = w^B (w^S)^{\top}\), as specified in Definition \ref{def:multiple-randomization}.
    \begin{equation}\label{eqn:w}
    \begin{array}{ccc}
    \W_1 =\left(\begin{array}{ccccc}
    0 & 0 & 0 & 0 & 0 \\
    1 & 1 & 1 & 1 & 1 \\
    0 & 0 & 0 & 0 & 0 \\
    1 & 1 & 1 & 1 & 1 
    \end{array}\right) &
    \W_2 =\left(\begin{array}{ccccc}
    1 & 0 & 1 & 0 & 0 \\
    1 & 0 & 1 & 0 & 0 \\
    1 & 0 & 1 & 0 & 0 \\
    1 & 0 & 1 & 0 & 0 
    \end{array}\right) &  
    \W_3 =\left(\begin{array}{ccccc}
    0 & 0 & 0 & 0 & 0 \\
    1 & 0 & 1 & 0 & 0 \\
    0 & 0 & 0 & 0 & 0 \\
    1 & 0 & 1 & 0 & 0 
    \end{array}\right) \\
    \text{Buyer-side randomization} & \text{Seller-side randomization} & \text{Two-sided randomization}
    \end{array}
    \end{equation}
\end{example}

In Example \ref{example:crd}, we demonstrate the concept of an independent two-sided experiment under complete randomization. In this design, treatment assignments for each population are uniformly distributed,  ensuring a fixed fraction of treated units. Another variant of the two-sided design assigns treatments based on independent and identically distributed (i.i.d.) Bernoulli trials for each buyer and seller. As discussed later, both designs enjoy a nice property known as ``design symmetry'' or ``exchangeability'', which motivates the use of permutation tests as a computationally straightforward approach.

\begin{example}[Two-Sided Design under Complete Randomization]\label{example:crd}

For fixed integers, $0 < I_1 < I$ and $0 < J_1 < J$, $I_0 = I - I_1$ and $J_0 = J - J_1$, let $p^B \sim \text{Unif}(\mathbb{W}^B) \text{ and }  p^S \sim \text{Unif}(\mathbb{W}^S)$ where $\mathbb{W}^B= \{w^B \in \{0,1\}^I: \sum_{i=1}^I w_i^B = I_1\}$ and $\mathbb{W}^S = \{w^S \in \{0,1\}^J: \sum_{j=1}^J w_j^S = J_1\}$. 

This induces a uniform two-sided randomized design as follows.
\begin{equation*}
p(\boldsymbol W)
=
\begin{cases}
\displaystyle
\bigl[C(I,I_1)\,C(J,J_1)\bigr]^{-1},
& \text{if } \boldsymbol W = w^B (w^S)^\top
\text{ for some } w^B\in\mathbb W^B,\; w^S\in\mathbb W^S,\\[1em]
0, & \text{otherwise},
\end{cases}
\end{equation*}
\end{example}

Next, we introduce the concept of \textit{local interference}, which forms the basis for the two null hypotheses of interest. Under the classical Stable Unit Treatment Value Assumption (SUTVA)~\citep{Rubin1972, Rubin1980}, the potential outcome can be represented by only the treatment status of a given pair, i.e. \(Y_{i,j}(W_{i,j})\). However, recent literature on online marketplace experiments~\citep{whymarketplace, toulis2016long, johari2021, imbens2021} has pointed out that this assumption is untenable as there is likely interference 
due to a single seller interacting with multiple buyers or a single buyer interacting with multiple sellers. On the other hand, without any restrictions on the potential outcome function, the analysis would be intractable as the treatment space is exponentially large. To make progress, we follow \cite{imbens2021} and assume that interference may only occur between the activities of a given buyer or seller in any given buyer-seller pair.

\begin{assumption}[Local interference]\label{ass:local-interference}
Potential outcomes satisfy $Y_{i,j}(\ww) = Y_{i,j}(\ww^\prime)$ for all pairs $(i,j)$ and assignments $\ww,\ww'$, provided the following conditions hold: (a) the assignments for the pair $(i, j)$ coincide, such that $w_{i,j}=w_{i,j}^\prime$; (b) the fraction of treated sellers for buyer $i$ coincide under $\ww$ and $\ww^\prime$, and (c) the fraction of treated buyers for seller $j$ coincide under $\ww$ and $\ww^\prime$.
\end{assumption}

Following Lemma 3.5 in \cite{imbens2021} and under Assumption \ref{ass:local-interference}, we can express the potential outcomes in terms of four distinct treatment exposure cases as follows:\footnote{
To be more specific, by Assumption \ref{ass:local-interference}, the potential outcomes can be represented as \(Y_{i,j}(W_{i,j}, \bar w^B_i, \bar w^S_j)\), where \(\bar w^B_i = \frac{1}{J} \sum_{j=1}^J W_{i,j}\) and \(\bar w^S_j = \frac{1}{I} \sum_{i=1}^I W_{i,j}\) are observed treated fractions. Since $\bar w^B_i = \frac{1}{J} \sum_{j=1}^J w_i^B w_j^S = w_i^B r^S$ and $\bar w^S_j = w_j^S r^B$, where $r^S$ and $r^B$ are the treated fractions of sellers and buyers, respectively, we have \(Y_{i,j}(W_{i,j}, \bar w^B_i, \bar w^S_j) = Y_{i,j}(w_i^B w_j^S, w_i^B r^S, w_j^S r^B)\). Therefore, the potential outcomes are 
a function of only $w_i^B$ and $w_j^S$ whenever our testing procedures do not alter the treated fractions. This justifies the simplified notation of Equation~\eqref{eqn:potential-outcomes}.
 }
\begin{equation}\label{eqn:potential-outcomes}
    Y_{i,j}(\W) = Y_{i,j}(w_i^B, w_j^S) =
    \begin{cases}
      Y_{i,j}(0,0) & \text{if both buyer and seller are untreated}, \\
      Y_{i,j}(1,0) & \text{if buyer is treated and seller is untreated}, \\
      Y_{i,j}(0,1) & \text{if buyer is untreated and seller is treated}, \\
      Y_{i,j}(1,1) & \text{if both buyer and seller are treated}.
    \end{cases}
\end{equation}

Certain constrasts between the above potential outcomes express spillover effects from either the buyer side or the seller side. For example, the buyer spillover effect examines the impact on a buyer-seller pair \((i,j)\) when buyers are treated versus when no buyers are treated, holding the seller untreated. Similarly, the seller spillover effect compares the outcomes for a seller when sellers are treated versus untreated, assuming the buyer remains untreated. In Hypotheses \ref{hypothesis:buyer-spillovers} and \ref{hypothesis:seller-spillovers}, which we define below, we formalize the sharp null hypotheses to test for the existence of spillover effects from both the buyer and seller sides. Rejecting these two spillover null hypotheses provides evidence for the existence of interference, suggesting that spillover effects likely exist in the given direction. For instance, if the buyer spillover effect is non-zero, it could indicate that treated buyers' overall shopping experience is influenced by interactions with treated sellers, causing changes in behavior even when shopping from untreated sellers.

\begin{hypothesis}[Buyer spillover]\label{hypothesis:buyer-spillovers}
    The null hypothesis of no buyer spillover effects is 
    \begin{equation*}
\Hbuyer: Y_{i,j}(0,0) = Y_{i,j}(1,0), \text{ for all } i=1,\dots,I, j=1,\dots,J.
    \end{equation*}
\end{hypothesis}
\begin{hypothesis}[Seller spillover]\label{hypothesis:seller-spillovers}
    The null hypothesis of no seller spillover effects is 
    \begin{equation*}
        \Hseller: Y_{i,j}(0,0) = Y_{i,j}(0,1), \text{ for all } i=1,\dots,I, j=1,\dots,J.
    \end{equation*}
\end{hypothesis}
\begin{hypothesis}[Total effect]\label{hypothesis:total}
    The null hypothesis of no total treatment effects is given as:
    \begin{equation*}
        \Htotal: Y_{i,j}(0,0) = Y_{i,j}(1,1), \text{ for all } i=1,\dots,I, j=1,\dots,J.
    \end{equation*}
\end{hypothesis}
The total treatment effect (Hypothesis \ref{hypothesis:total}) aims to capture the difference between pairs where both the buyer and the seller are treated and pairs where neither is treated. 
We call this a ``total effect'' because it encompasses both spillover and direct treatment effects. For a clearer understanding, we discuss this concept through a linear outcome model in Example \ref{example:linear-outcome} that follows.

\begin{example}\label{example:linear-outcome}
Consider a linear additive model for the potential outcomes (with independent errors):
\begin{equation*}
    Y_{i,j}(\W) = \alpha w_{i}^B w_j^S + \beta w_i^B + \gamma w_j^S + \epsilon_{i,j} ~.
\end{equation*}
The null hypotheses in \ref{hypothesis:buyer-spillovers} to \ref{hypothesis:total} can be re-stated as 
$\Hbuyer: \beta = 0$, $\Hseller: \gamma = 0$, and 
$\Htotal: \alpha + \beta + \gamma = 0$. Here, $\beta$ and $\gamma$ denote the buyer and seller spillover effects, respectively, while $\alpha$ is the direct treatment effect.
\end{example}

Finally, we introduce Assumption \ref{ass:exchangeable}, which imposes symmetry on the experimental design. This symmetry allows us to construct permutation tests that are computationally efficient, but 
it is not necessary to construct valid randomization procedures in general.

\begin{assumption}[Design Symmetry]\label{ass:exchangeable}
    The probability distributions of treatment assignments for buyers, $p^B$, and sellers, $p^S$, are exchangeable. Specifically, $p^B(w_1^B, \dots, w_I^B) = p^B(w_{\pi_B(1)}^B, \dots, w_{\pi_B(I)}^B)$ for all permutations $\pi_B: [I] \rightarrow [I]$ and similarly $p^S(w_1^S, \dots, w_J^S) = p^S(w_{\pi_S(1)}^S, \dots, w_{\pi_S(J)}^S)$ for all permutations $\pi_S: [J] \rightarrow [J]$.
\end{assumption}


\newcommand{\mU}{\mathcal U}
\newcommand{\mW}{\mathcal{W}}
\newcommand{\mC}{\mathcal{C}}
\section{Main Method}\label{sec:randomization-tests}

\subsection{Overview of Conditional Randomization Tests}

The classical Fisher Randomization Test simulates the sampling distribution of the test statistic according to the actual treatment variation in the experiment~\citep[Chapter 5]{imbens2015causal}. 
While this procedure is valid for testing the sharp null hypothesis under interference, it is generally not valid for non-sharp hypotheses, which do not specify the complete schedule of potential outcomes.
The null hypotheses defined in the previous section fall into this category. For instance, 
if we observe outcome $Y_{ij}(1, 0)$ for a treated buyer $i$ and control seller $j$, then 
we cannot impute the outcome $Y_{ij}(0,0)$ under $\Hseller$---this null hypothesis is not sharp in the strict sense.

To address this issue, recent literature has proposed the use of conditional randomization tests that execute the resampling procedure on a subset of units, $\mU \subseteq \mathbb{U}$
and a subset of assignments, $\mW \subseteq \mathbb{W}$, such that the potential outcomes become fully specified~\citep{Aronow2012, Athey2018, Basse2019}. Note also that a unit is a buyer-seller pair in our case. In the terminology of~\citet{Basse2019}, $\mC = (\mU, \mW)$ is the \textit{conditioning event}, 
which may depend on the observed treatment assignment in a random way. Let $p(\mC \mid \W)$ denote its distribution, which is under the analyst's control. The conditional Fisher Randomization Test then randomizes treatment according to its conditional distribution:
\begin{equation}\label{eq:condFRT}
    p(\W \mid  \mC) \propto p(\mC \mid \W) p(\W).
\end{equation}

Such a test is finite-sample valid even for a non-sharp null hypothesis, provided that the conditioning event, $\mC$, is constructed 
in a way such that the potential outcomes for all units and assignments in $\mC$ can be imputed under the null hypothesis~\citep[Theorem 1]{Basse2019}. 
Sampling from the conditional distribution in~\eqref{eq:condFRT} may be challenging in general. 
However, under certain conditions on the design, $p(\W)$, and the conditioning mechanism, $p(\mC|\W)$, 
the distribution in~\eqref{eq:condFRT} can be reduced to a permutation distribution that is easy to sample from. We will consider this approach in the following sections.

\subsection{Testing spillover effects}

In this section, we develop conditional randomization tests for the 
spillover hypotheses $\Hbuyer$ and \(\Hseller\). Under symmetric designs, our randomization procedure entails straightforward permutations of treatment assignments on one side (buyer or seller), conditioned on the assignment of the other side (seller or buyer, respectively).  We begin with the buyer spillover hypothesis and describe the associated randomization procedure in detail.  The analysis of the
seller spillover hypothesis will be brief as it is completely symmetrical to the buyer case. 

For the buyer spillover hypothesis~\ref{hypothesis:buyer-spillovers}, 
we propose to condition on control sellers and fix those sellers to control across randomizations to ensure that the potential outcomes can be imputed.  Moreover, we will condition on the marginal number of treated buyer-sellers observed in the sample. Under Assumption~\ref{ass:exchangeable}, the conditional randomization distribution is exchangeable, resulting in a permutation test that can be efficiently implemented.

Formally, we define the following conditioning events:
\begin{equation}
     \mathcal{U} = \{(i,j): w^{obs, S}_j = 0\} \text{ and } \mathcal{W} = \{v (w^{obs, S})^{\top}: v \in  \mathbb{W}^B\} \cap \mathcal{M}~,\label{eq:cond_buyer}
\end{equation}
where $\mathcal{M} = \{v u^\top: v\in \mathbb{W}^B, u \in \mathbb{W}^S \text{ s.t. } \sum_{i=1}^I v_i = \sum_{i=1}^I w^{obs, B}_i, \sum_{j=1}^J u_j = \sum_{j=1}^J w^{obs, S}_j\}$. Here, we use the superscript `obs' to emphasize that we are referring to the actual treatments realized in the observed data.
Thus, $\mU$ denotes the buyer-seller pairs for which the seller part of the pair is assigned control. Set $\mW$ denotes the assignments that these sellers in control, while $\mathcal{M}$ is the set of assignments that match the marginal number of treated buyers and sellers observed in the sample. We note that all these sets are random as they all depend on the observed treatment assignment vector, $\Wobs$. We are now ready to define the main randomization test for the buyer spillover hypothesis.

\begin{procedure}[Test for Buyer Spillover Effect $\Hbuyer$]\label{procedure:spillover}
Let $\mC = (\mU, \mW)$ denote the conditioning event implied by the definitions in Equation~\eqref{eq:cond_buyer}.
\begin{enumerate}
    \item Compute the test statistic $T(\W^{obs}\mid \mathbf{Y}^{obs}, \mathcal{C})$ under the observed treatment assignment. For example, we can use the difference-in-means,
    \begin{equation*}
        T(\W \mid \mathbf{Y}, \mathcal{C}) = \frac{1}{n_{1}} \sum_{(i,j)\in\mU} Y_{i,j} (1-w^{S}_j) w^{B}_i - \frac{1}{n_{0}}  \sum_{(i,j)\in\mU} Y_{i,j} (1-w^{S}_j) (1- w^{B}_i) ~,
    \end{equation*}
    where $n_{1} =  \sum_{(i,j)\in\mU} (1-w^{S}_j) w^{B}_i, \quad n_{0} =  \sum_{(i,j)\in\mU} (1-w^{S}_j) (1-w^{B}_i).$
    \item 
    For \( l = 1, 2, \dots, L \), generate
    \( \W^{(l)} = w^{B,(l)} (w^{S, obs})^\top \) where \( w^{B,(l)} \) is a random permutation of \( w^{B, obs} \), and compute \( T(\W^{(l)} \mid \mathbf{Y}^{obs}, \mathcal{C}) \).
    
    \item Reject $\Hbuyer$, denoted by $\phi(\W^{obs}; \mathcal{C})$, if
    \begin{equation*}
        \frac{1}{L+1} 
        \left[1+\sum_{l=1}^{L} \mathbb{I}\left\{T(\W^{(l)}\mid \mathbf{Y}^{obs}, \mathcal{C}) \geq T(\W^{obs}\mid \mathbf{Y}^{obs}, \mathcal{C}) \right\} \right] \leq \alpha.
    \end{equation*}
\end{enumerate}
\end{procedure}

Figure~\ref{fig:buyer} is a graphical depiction of Procedure~\ref{procedure:spillover} for testing the 
buyer spillover effect. The shaded parts of the array in the figure illustrates 
that the test conditions on the set of sellers that were not treated. Procedure~\ref{procedure:spillover} then randomizes the treatments of buyers (i.e., rows in the array). Under certain assumptions, such randomization is equivalent to permutations as explained in the following remark.

\begin{figure}[h!]
    \centering
    \includegraphics[width=0.65\textwidth]{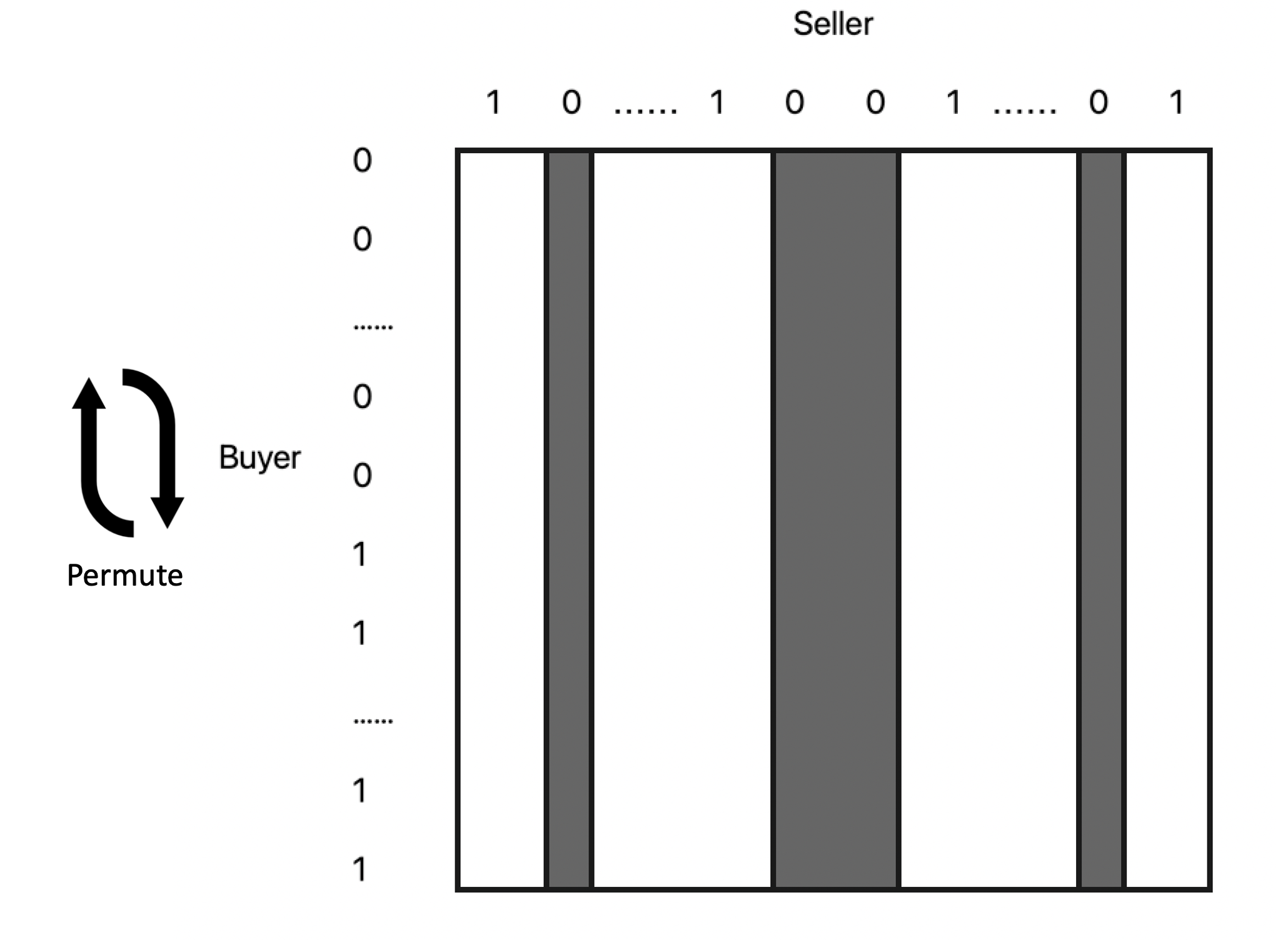}
    \caption{Graphical illustration of the conditioning event for Procedure~\ref{procedure:spillover}. 
    The shaded area illustrates that the procedure conditions on control sellers (columns), and then randomizes the treatments on buyers (rows).}
    \label{fig:buyer}
\end{figure}

\begin{remark}\label{remark:proof-explain}
Under Assumption \ref{ass:exchangeable}, we can show that the permutation of the buyer treatments in step 2 leads to a valid permutation test by extending the equivariance result of~\cite{basse2024}. In particular, following the definitions in Equation~\eqref{eq:cond_buyer}, we can show that a permutation in the assignment of buyer-seller treatment
does not affect the selected buyer-seller pairs in $\mU$. 
Moreover, such permutation also implies an analogous permutation of the treatment exposures of these pairs, a property known as equivariance. We give the full details of the proof in the Appendix.
\end{remark}

\begin{remark}[Randomization tests when design symmetry fails]\label{remark:sample-conditional}
If Assumption \ref{ass:exchangeable} does not hold, the permutation of the buyer treatments in step 2 no longer leads to a valid randomization test. 
In such settings, we need to modify step 2 of Procedure \ref{procedure:spillover} as follows:
    \begin{itemize}
        \item[2'.] For $l=1,2,\dots, L$, draw $\W^{(l)} \sim p(\W\mid \mathcal{C})$ and compute $T(\W^{(l)}\mid \mathbf{Y}^{obs}, \mathcal{C})$. Here, $\W^{(l)}$ is defined as $w^{B, (l)} (w^{S, obs})^\top$, where
        \begin{equation*}
            w^{B, (l)} \sim p^B\left(w \mid \sum_{i=1}^I w_i = \sum_{i=1}^I w_i^{B, obs}\right)~.
        \end{equation*}
    \end{itemize}
     That is, we repeatedly sample from the buyer-side design, conditioning on the treated fraction in the observed buyer-side assignment. However, sampling from such conditional distribution---e.g., through rejection sampling---may be computationally challenging depending on the particular experimental design. 
     See  Appendix \ref{app:details} for more details and the proof.
\end{remark}
\begin{remark}[Seller spillover hypothesis]
    Testing for the seller spillover hypothesis is completely analogous to Procedure~\ref{procedure:spillover}. The idea is to condition on control buyers, and change the definitions in Equation~\eqref{eq:cond_buyer} into
    \begin{equation*}
    \mathcal{U} = \{(i,j): w^{obs, B}_i = 0\} \text{ and } \mathcal{W} = \{w^{obs, B} u^{\top}: u \in  \mathbb{W}^S\} \cap \mathcal{M}.
\end{equation*}
    Operationally, we may simply transpose the buyer-seller array and execute Procedure~\ref{procedure:spillover}.
\end{remark}


The validity of our proposed randomization tests in Procedure \ref{procedure:spillover} follows by adapting Theorem 2 of \cite{puelz2021} in the setting of two-sided experiments. We provide the statement of the validity result in the following theorem.

\begin{theorem}\label{thm:validity1}
Consider an independent two-sided randomized design (Definition \ref{def:multiple-randomization}) where Assumption \ref{ass:local-interference} holds. Let $\mathcal{C}$ be a buyer spillover conditioning event, as defined in Equation~\eqref{eq:cond_buyer}.
\begin{enumerate}
    \item Under Assumption \ref{ass:exchangeable}, the testing procedure $\phi(\W^{obs}; \mathcal{C})$, as defined in Procedure \ref{procedure:spillover}, is valid. Specifically, for any significance level $\alpha \in (0,1)$, the expectation $E[\phi(\W^{obs}; \mathcal{C})\mid\mathcal{C}] \leq \alpha$ holds under the null hypothesis $\Hbuyer$ (or $\Hseller$).    \item If Step (b) of Procedure \ref{procedure:spillover} is replaced with the one in Remark \ref{remark:sample-conditional}, the validity of $\phi(\W^{obs}; \mathcal{C})$ is maintained without assuming Assumption \ref{ass:exchangeable}.
\end{enumerate}
\end{theorem}

\subsection{Testing Total effects}

\subsubsection{Testing Procedure with $k$-Block Conditioning Event}
In this section, we proceed to discuss the testing procedure for total effects. We propose a randomization procedure that relies on a \textit{\(k\)-Block Conditioning Event}. 
The idea is to split focal units into \(k\)-by-\(k\) blocks, with the distinctive arrangement ensuring that units within a specific block do not share rows or columns with units from other blocks. Furthermore, every unit within a block shares the same treatment status for both buyers and sellers. 

Figure \ref{fig:total} illustrates an example of a $k$-block conditioning event corresponding to the shaded diagonal blocks. The treatment assignments within each block are either $(w_i^B, w_j^S) = (0,0)$ or $(w_i^B, w_j^S) = (1,1)$, with no overlap of blocks across columns or rows. Such arrangement is important as it allows the use of efficient permutation procedures under symmetric two-sided randomized designs. In other words,
under symmetric treatment assignment designs, such as complete randomization and Bernoulli trials, our proposed tests involve permuting the treatment assignments of focal units across blocks.

\newcommand{\IK}{\mathcal{I}^{(k)}(\Wobs)}
\newcommand{\JK}{\mathcal{J}^{(k)}(\Wobs)}
\newcommand{\UK}{\mathcal{U}^{(k)}(\Wobs)}
\newcommand{\WK}{\mathcal{W}^{(k)}(\Wobs)}
\newcommand{\CK}{\mathcal{C}^{(k)}}

Towards testing the total effect,  
let $\IK$ denote a random partition of the buyer set $[I]$ into $k$ equal-sized non-overlapping subsets with all buyers having the same treatment status within each subset.\footnote{We assume $I$ and $\sum_{i=1}^Iw_i^{obs, B}$ are both divisible by $k$ for simplicity. When this is not the case, our methods remain valid, and the framework can be extended accordingly, though at the cost of more cumbersome notation.} That is,
\begin{equation}\label{eqn:k-partition}
\IK= \{ \mathcal{I}_s : 1\leq s \leq I/k\},
\end{equation}
such that for all $1 \leq s, s^\prime \leq I/k$: 
\begin{itemize}
    \item[(i)] $|\mathcal{I}_s|=k$, and $\cup_s~\mathcal{I}_s = [I]$,  and $\mathcal{I}_s \cap \mathcal{I}_{s^\prime} = \emptyset$.
    \item[(ii)] $w_i^{obs, B} = w_{i'}^{obs,B}$ for all $i,i' \in \mathcal{I}_s$.
\end{itemize}
Similarly, we define $\JK $ as a random partition of the seller set $[J]$ into $k$-sized subsets with 
sellers having the same treatment status within each subset. Next, define
\begin{equation}\label{eq:Uk}
\UK = \bigcup_{ 1 \leq s \leq \min(I/k, J/k)} \{ (i, j) : w_i^{obs,B} = w_j^{obs, S}, i \in \mathcal{I}_s, j\in\mathcal{J}_{s} \},
\end{equation}
where $\mathcal{I}_s, \mathcal{J}_{s}$ are defined above.
That is, $\UK$ is the collection of buyer-seller pairs constructed from the cross-product of $\IK$ and $\JK$, making sure that all buyers and sellers have the same treatment status in the group. Figure~\ref{fig:total} below illustrates this construction. In the figure, the set $\UK$ therefore corresponds to the diagonal blocks 
in the shaded area of the buyer-seller array. Note how these blocks are non-overlapping, 
they partition the sets $[I]$ and $[J]$, and, within each block, the individual treatments of every buyer and seller are identical.

\begin{figure}[h]
    \centering
    \includegraphics[width=0.57\textwidth]{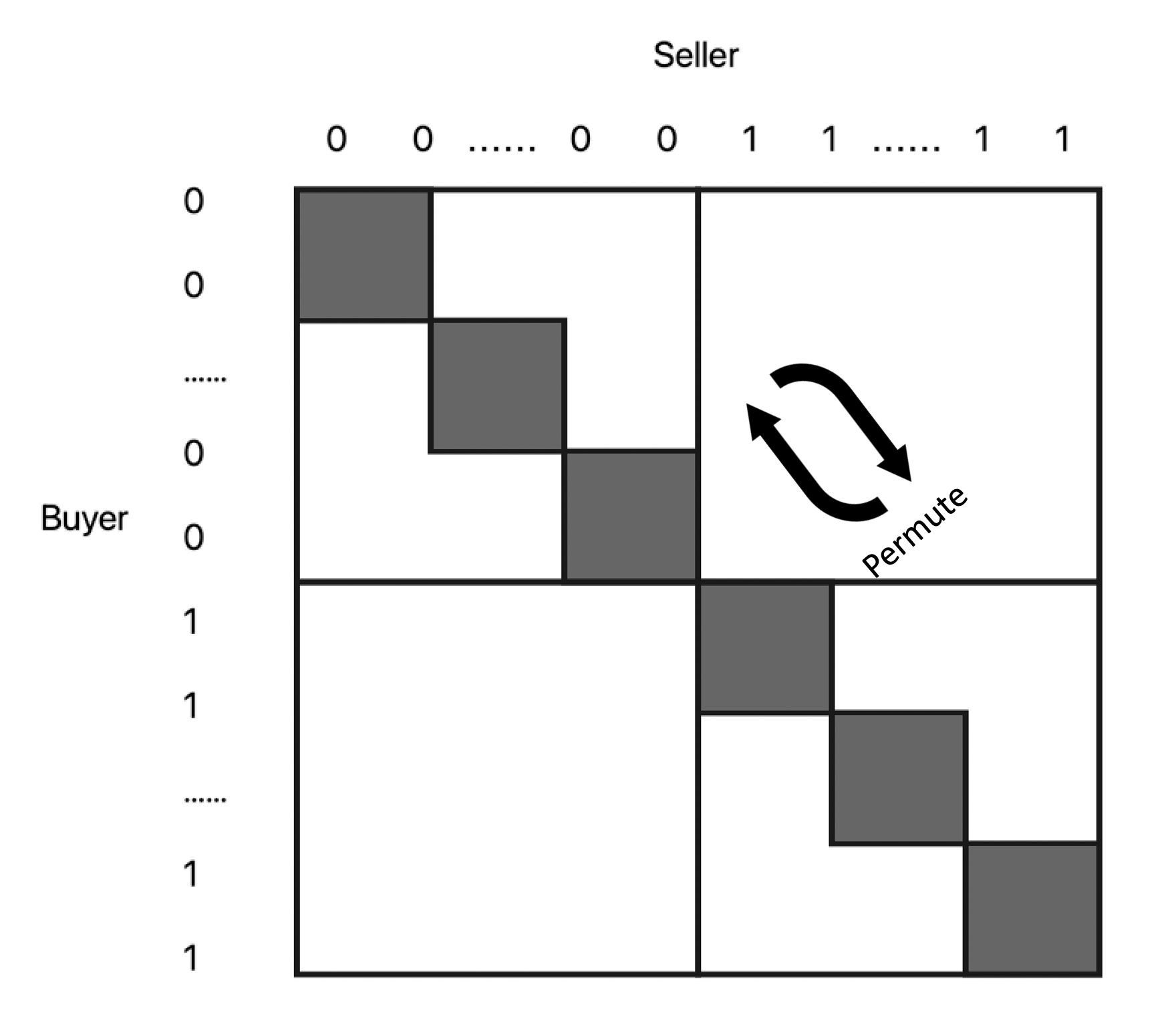}
    \caption{Graphical illustration of the conditioning event for Procedure~\ref{procedure:total}. The shaded area illustrates that the procedure permutes the treatment assignments across diagonal blocks.}
    \label{fig:total}
\end{figure}

\begin{definition}[$k$-Block Conditioning Event]\label{def:k-block}
Given treatment assignment $\Wobs$, a $k$-block conditioning event 
is defined as $\CK = (\UK, \WK)$, 
where $\UK$ is defined in Equation~\eqref{eq:Uk} and $\mathcal{W}^{(k)}$ is 
a subset of treatment assignments defined as follows:
\begin{align*}
    \WK = \big \{ \W: w_i^B  =  w_j^S \text{ for all } (i,j)\in \UK  \big \} \cap \mathcal{M}~,
\end{align*}
where $\mathcal{M}$ is defined in Equation~\eqref{eq:cond_buyer} as the set of assignments that preserve the treated fraction in each population.
\end{definition}

We are now ready to define our main randomization procedure for testing the total null hypothesis.

\begin{procedure}[Test Total Effect $\Htotal$]\label{procedure:total}
Let $\CK = (\UK, \WK)$ denote the $k$-conditioning event from Definition~\ref{def:k-block}.
\begin{enumerate}
    \item Compute test statistics $T(\W^{obs}\mid \mathbf{Y}^{obs}, \CK)$ under the observed treatment assignment. For example, we can use the difference-in-means:
    \begin{equation*}
        T(\W\mid \mathbf{Y}, \CK) = \frac{1}{n_{1,k}} \sum_{(i,j)\in\UK} Y_{i,j} w^{S}_j w^{B}_i  - \frac{1}{n_{0,k}}  \sum_{(i,j)\in\UK} Y_{i,j} (1-w^{S}_j) (1- w^{B}_i) ~,
    \end{equation*}
    where $n_{1,k} =  \sum_{(i,j)\in\UK} w^{S}_j w^{B}_i  , \quad n_{0,k} =  \sum_{(i,j)\in\UK} (1-w^{S}_j) (1-w^{B}_i) .$
    \item For $l=1,2,\dots, L$, obtain $\W^{(l)} = w^{B,(l)} (w^{S, (l)})^\top$ where $w^{B,(l)},  w^{S,(l)}$ are block-wise random permutation of $w^{B,obs},  w^{S,obs}$ given by
    \begin{align*}
        w^{B,(l)}
        &:= \Bigl\{ w^B \in \mathbb{W}^B \ \Big|\ 
        (w^B_i)_{i\in \mathcal{I}_{s}} = (w_i^{B,\mathrm{obs}})_{i\in \mathcal{I}_{\pi(s)}}
        \ \ \forall s\in\{1,\dots,\min( I/k , J/k  )\}
        \Bigr\}, \\
        w^{S,(l)}
        &:= \Bigl\{ w^S \in \mathbb{W}^S \ \Big|\ 
        (w^S_j)_{j\in \mathcal{J}_{s}} = (w_j^{S,\mathrm{obs}})_{j\in \mathcal{J}_{\pi(s)}}
        \ \ \forall s\in\{1,\dots,\min( I/k , J/k  )\}
        \Bigr\}.
    \end{align*}
    where $\pi$ is a permutation on $\{1,\dots,\min( I/k ,  J/k ) \}$. Compute the randomized statistic $T(\W^{(l)}\mid \mathbf{Y}^{obs}, \CK)$.
    \item Reject $\Htotal$, denoted by $\phi(\W^{obs};\CK)=1$, if
    \begin{equation*}
        \frac{1}{L+1}\left[1+\sum_{l=1}^{L} \mathbb{I}\left\{T(\W^{(l)}\mid \mathbf{Y}^{obs}, \CK) \geq T(\W^{obs}\mid \mathbf{Y}^{obs}, \CK) \right\} \right] \leq \alpha.
    \end{equation*}
\end{enumerate}
\end{procedure}

Intuitively, to test $\Htotal$, Procedure~\ref{procedure:total} conditions 
on a $k$-block conditioning event ---e.g., the $k$ diagonal blocks shown in Figure~\ref{fig:total}--- and permutes the block treatments $(0,0)$, $(1,1)$ at the block level. We note that the diagonal structure in the blocks is not essential and alternative constructions can be valid as long as the blocks are non-overlapping in neither the buyer nor the seller dimension.

The following result shows that Procedure~\ref{procedure:total} is valid 
under randomized designs that satisfy local interference and design symmetry.

\begin{theorem}\label{thm:validity2}
Consider an independent two-side ranomdized design (Definition \ref{def:multiple-randomization}) where Assumption \ref{ass:local-interference} holds. Let $\CK$ be a $k$-block conditioning event  as in Definition \ref{def:k-block}.
Suppose also that Assumption \ref{ass:exchangeable} holds. 
Then, $\phi(\W^{obs}; \CK)$ defined in Procedure  \ref{procedure:total} is finite-sample valid, such that $\mathbb{E}[ \phi(\Wobs; \CK) \mid \CK ] \le \alpha$ for any finite $I,J$.
\end{theorem}


In the case of non-symmetric designs where Assumption~\ref{ass:exchangeable} does not hold, we can still construct valid randomization tests, albeit in a more complex form.
Specifically, suppose without loss of generality that $I =J$ and the $k$-partitioning on two populations are identical, i.e. 
$\IK = \JK$. Then, the permutation in step 2 of Procedure~\ref{procedure:total} with sampling from the conditional
    \begin{equation}\label{eq:conditional}
        w^{B, (l)} \sim p^B\left(w \mid \sum_{i=1}^I w_i^B = \sum_{i=1}^I w_i^{B, obs}, w_i^B=w_j^B \text{ for all } i,j \in \mathcal{I}_s \text{ for all } 1\leq s \leq I/k\right).
    \end{equation}
    and setting  $\W^{(l)} = w^{B, (l)} (w^{B, (l)})^\top$. 
Despite its simplicity, this approach may be complex to implement since sampling from the conditional~\eqref{eq:conditional} could be computationally challenging under arbitrary non-symmetric designs.
    


\subsubsection{Selecting $k$ Based on Power Considerations}\label{sec:select-k}

Theorem \ref{thm:validity2} establishes that Procedure \ref{procedure:total} is valid for any block size \(k\). However, the statistical power 
of the procedure likely depends on a complex trade-off relating to the value of $k$. 
A smaller value of $k$ leads to using more blocks and thus is beneficial for the test's power since it increases the support of the randomization distribution. 
However, it also leads to smaller block sizes, and thus it may decrease power due to the reduced sample size. Conversely, a higher value of $k$ leads to the reverse effect. 

To illustrate further, consider a completely randomized two-sided design following Example \ref{example:crd}, with \(I_1\) and \(J_1\) treated rows and columns, respectively. 
For the sake of simplicity, let's assume the observed data matrix is square with half of the rows and columns assigned to the treatment group, such that \(I = J = 2n\) and \(I_1 = I_0 = J_1 = J_0 = n\), with \(n/k\) being an integer. Under these conditions, the total number of unique treatment assignments ---i.e., the support of the randomization distribution--- in a $k$-block conditioning event equals \(C(2n/k, n/k)=O(2^{2n/k}/\sqrt{2n/k})\), where \(C(\cdot, \cdot)\) represents the binomial coefficient. Moreover,  the sample size of focal units is $|\UK| = 2nk$. Thus, as $k$ increases, the sample size increases but the support of the randomization distribution decreases, which 
have opposing effects on the test power. Conversely, as $k$ decreases, 
the total sample size decreases but the support of the randomization distribution increases.



 


To navigate this trade-off, we can leverage the results of \cite{puelz2021}, who 
established a lower bound for the power of general conditional randomization tests under some plausible assumptions, which we detail in Appendix \ref{app:details-power}. Let \(\phi_{n,k}\) denote the power of a conditional randomization test in our setting. The following proposition provides the bound under complete randomization, as in Example \ref{example:crd}:
\begin{proposition}\label{prop:power}
Under the assumptions of Theorem 3 (Appendix \ref{app:details-power}) in \cite{puelz2021} and 
the complete randomization design of Example \ref{example:crd}, the power of Procedure~\ref{procedure:total} satisfies:
\begin{equation*}
\phi_{|\mathcal{U}|,|\mathcal{W}|} \geq \frac{1}{1 + A e^{-a \tau \sqrt{2nk}}} - O\left(C(2n/k, n/k)^{-0.5 + \delta}\right) - \epsilon,
\end{equation*}
for any small \(\delta\) and sufficiently large \(|\mathcal{W}|\), where \(A, a > 0\) are constants.
\end{proposition}

\begin{figure}[h]
    \centering
    \includegraphics[width=0.6\textwidth]{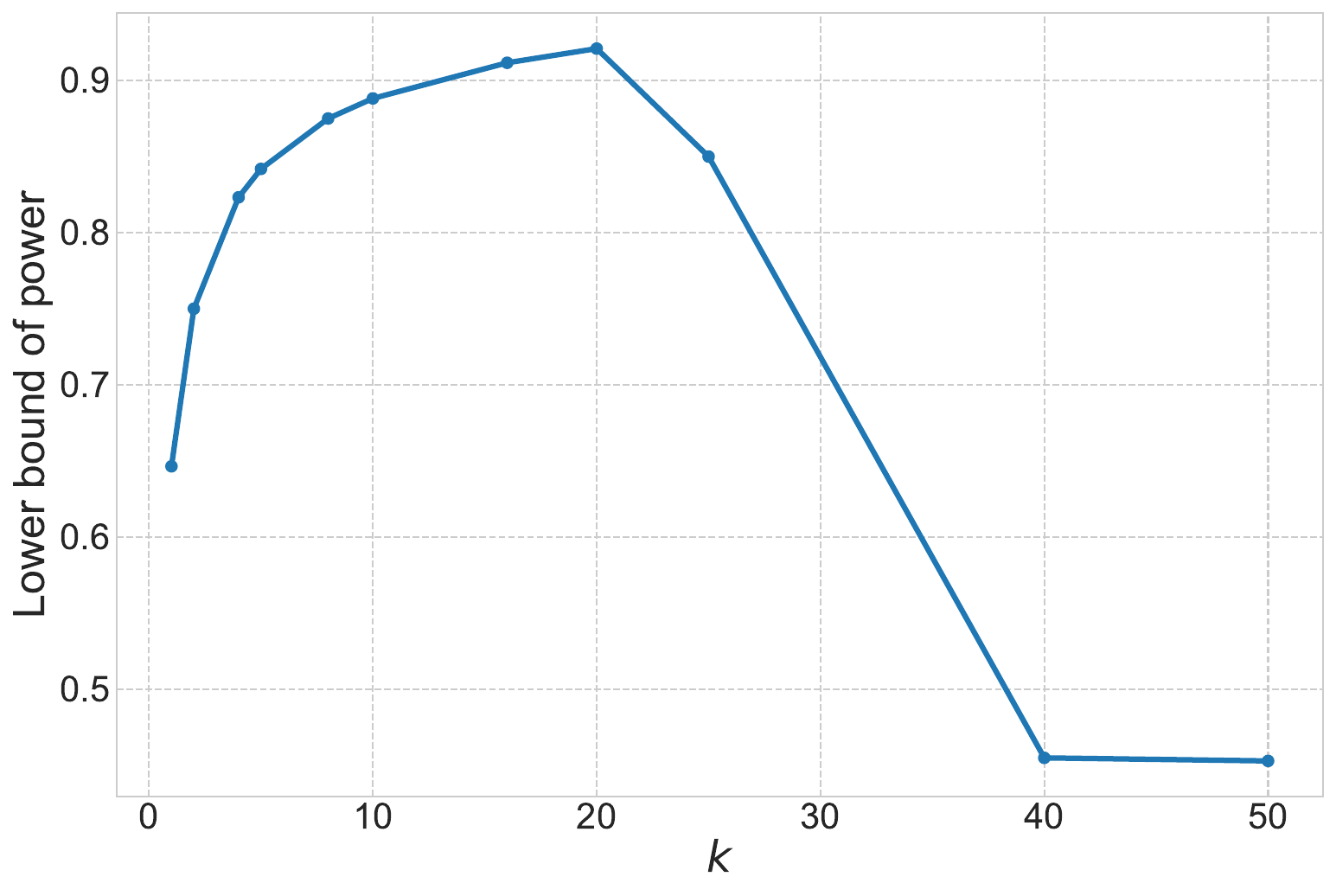}
    \caption{Graphical illustration of the lower bound on power for varying \( k \) with \( n = 400 \), \( A = 1 \), and \( a = \tau = 0.01 \).}
    \label{fig:power}
\end{figure}

In Figure \ref{fig:power}, we illustrate the lower bound on power as a function of the block size \( k \), highlighting the trade-off between sample size and the number of unique treatment assignments as \( k \) increases. While the plot identifies the optimal \( k \) that maximizes power, in practice, the parameters of the lower bound function are unknown. Consequently, the optimal \( k \) cannot be directly determined using Proposition \ref{prop:power}. However, a heuristic approach can be used to select the largest feasible \( k \) based on a predetermined maximum power threshold. Based on Proposition \ref{prop:power}, $O\left(C(2n/k, n/k)^{-0.5 + \delta}\right)$ controls the maximum power of the test, while $2nk$ determines how rapidly the power function reaches its maximum (sensitivity). Therefore, to optimize the detection of treatment effects, we recommend selecting the largest feasible block size $k$ based on a predetermined maximum power. For example, setting this power to 0.95 ensures the maximal utilization of observations to detect the minimal size of the treatment effect. In practice, we can choose $k$ to satisfy $C(2n/k,n/k)\approx\frac{1}{(1-\beta)^2}$, where $\beta$ is the predetermined power level. For example, consider the completely randomized two-sided design discussed before, where choosing a block size $k = n/6$ results in $|\mathcal{W}| = C(12,6) = 924$, achieving a maximum power approximately equal to 0.967.

\section{Randomization Tests for Weak Null Hypotheses}\label{sec:weak-null}
There has been a growing number of works on the study of randomization tests under weak null hypotheses; see, for example, \citep{chung2013,Ding2017,Canay2017,diciccio2017robust,DiCiccio&Romano2017,ZHAO2021278,Wu2021}. 
This literature has pointed out that randomization tests based on sharp nulls are not necessarily valid under the weak null. Nevertheless, in classic one-sided experiments, validity can often be restored by studentizing the test statistic. The most commonly used approach relies on a Neyman-style variance estimator, which coincides with the usual two-sample $t$-test variance.

In this section, we study use of conditional randomization tests with suitable test statistics and studentization for testing weak null hypotheses in two-sided randomization design. We document that the validity of these tests hinges on both the specific form of the null as well as the way in which studentization is carried out. 

Specifically, we consider two types of weak nulls. The first, which we call the \textit{population average spillover effect from buyers}, imposes a single mean-zero restriction on the average treatment effect after pooling over all buyers and sellers.
\begin{hypothesis}[Population Average Spillover Effect from Buyers]\label{hypothesis:weak-buyer-spillovers1}
\begin{equation}\label{eq:weak_null1}
H_0^{wb,1}: \frac{1}{IJ} \sum_{i=1}^{I}\sum_{j=1}^J Y_{i,j}(0,0) = \frac{1}{IJ} \sum_{i=1}^{I}\sum_{j=1}^J Y_{i,j}(1,0)~.
\end{equation}
\end{hypothesis}
The second, which we call the \textit{seller-specific average spillover effect from buyers}, imposes a collection of mean-zero restrictions, one for each seller, requiring that the buyer-level average treatment effect vanish seller by seller.
\begin{hypothesis}[Seller-specific Average Spillover Effect from Buyers]\label{hypothesis:weak-buyer-spillovers2}
\begin{equation}\label{eq:weak_null2}
H_0^{wb,2}: \frac{1}{I} \sum_{i=1}^{I} Y_{i,j}(0,0) = \frac{1}{I} \sum_{i=1}^{I} Y_{i,j}(1,0) \quad \text{for all } j~.
\end{equation}
\end{hypothesis}

We show that the natural test statistic, together with commonly used Neyman-style studentization, leads to a valid test for the seller-specific null, but not for the population null. Therefore, we propose a modified, two-way (buyer $\times$ seller) variance estimator that restores asymptotic validity for the population-average weak null by absorbing the additional sampling variation induced by seller randomization.




\subsection{Standard Neyman-style studentization}\label{subsec:weak-null-standard}

We begin by presenting the testing procedure under standard Neyman-style studentization:

\begin{procedure}[Test Weak Null via standard (buyer-side) studentization]\label{procedure:weak-spillover}
Let $\mC = (\mU, \mW)$ denote the conditioning event implied by the definitions in Equation~\eqref{eq:cond_buyer}.
\begin{enumerate}
    \item Compute the studentized test statistics $T^{WB}(\W^{obs}\mid \mathbf{Y}^{obs}, \mathcal{C})$ under the observed treatment assignment,
    \begin{equation*}
        T^{WB}(\W\mid \mathbf{Y}, \mathcal{C}) = \frac{T^{B}(\W\mid \mathbf{Y}, \mathcal{C})}{\sqrt{V^{B}(\W\mid \mathbf{Y}, \mathcal{C})} } ~,
    \end{equation*}
    where $T^{B}(\W \mid \mathbf{Y}, \mathcal{C})$ is identical to $T(\W \mid \mathbf{Y}, \mathcal{C})$ in Procedure \ref{procedure:spillover} and 
    \begin{equation*}
        V^{B}(\W\mid \mathbf{Y}, \mathcal{C}) = \frac{s_1^2}{I_1} + \frac{s_0^2}{I_0}
    \end{equation*}
    with $s_1^2 = \frac{1}{I_1-1} \sum_{i=1}^{I}w_i^B\left( \bar{Y}_i^B - \hat{Y}(1) \right)^2$ and $s_0^2=\frac{1}{I_0-1}\sum_{i=1}^{I}(1-w_i^B)\left( \bar{Y}_i^B - \hat{Y}(0) \right)^2$, and sample mean $\hat{Y}(z) =  \sum_{i:w_i^B=z} \bar{Y}_i^B/I_z$ and $I_z = \sum_{i=1}^I \mathbf{I}\{ w_i^B=z\}$ for $z \in \{0,1\}$, and $\bar{Y}_i^B = \frac{1}{J_0} \sum_{j=1}^J Y_{i,j} (1-w_j^S)~.$
    \item For \( l = 1, 2, \dots, L \), generate
    \( \W^{(l)} = w^{B,(l)} (w^{S, obs})^\top \) where \( w^{B,(l)} \) is a random permutation of \( w^{B, obs} \), and compute \( T^{WB}(\W^{(l)} \mid \mathbf{Y}^{obs}, \mathcal{C}) \).
    \item Reject $H_0^{wb}$ if
    \begin{equation*}
        \frac{1}{L+1}\left[ 1 +\sum_{l=1}^{L} \mathbb{I}\left\{T^{WB}(\W^{(l)}\mid \mathbf{Y}^{obs}, \mathcal{C}) \geq T^{WB}(\W^{obs}\mid \mathbf{Y}^{obs}, \mathcal{C}) \right\} \right] \leq \alpha.
    \end{equation*}
\end{enumerate}
\end{procedure}


The following theorem summarizes the asymptotic properties of Procedure~\ref{procedure:weak-spillover} 
under the two weak null hypotheses. Importantly, it establishes (i) asymptotic validity under the seller-specific weak null $H_0^{wb,2}$ in \eqref{eq:weak_null2}, and (ii) the failure of asymptotic validity in general under the population-average weak null $H_0^{wb,1}$ in \eqref{eq:weak_null1}.

\begin{theorem}\label{thm:validity-weak}
Consider an independent two-sided randomized design under complete randomization where Assumption \ref{ass:local-interference} and regularity conditions \ref{assump:bounded-fourth-moments}-\ref{assump:uniform-integrability} hold. 
For any level $\alpha \in (0,1)$, the testing procedure in Procedure \ref{procedure:weak-spillover} satisfies:
\begin{equation*}
    \lim_{I, J \rightarrow \infty} P\left\{E\left[ \mathbb{I}\left\{T^{WB}(\W\mid \mathbf{Y}^{obs}, \mathcal{C}) \geq T^{WB}(\W^{obs}\mid \mathbf{Y}^{obs}, \mathcal{C}) \right\} \right] \leq \alpha\right\} \leq \alpha~,
\end{equation*}
under the null hypothesis $H_0^{wb,2}$ in (\ref{eq:weak_null2}), where the expectation is with respect to $p(\W\mid \mathcal{C})$. Conversely, asymptotic validity does not hold in general under the null hypothesis $H_0^{wb,1}$ in (\ref{eq:weak_null1}).
\end{theorem}

The divergence in the asymptotic validity results between $H_0^{wb,1}$ and $H_0^{wb,2}$ can be understood through the standard comparison between the randomization and sampling distributions of the studentized test statistic. In the literature on randomization tests for weak null hypotheses, validity is typically established by showing that the randomization distribution---here, the conditional permutation distribution induced by $p(\W \mid \mathcal{C})$---asymptotically stochastically dominates the sampling distribution induced by the original design. A common route is to show that the randomization distribution converges to $N(0,1)$, while the sampling distribution converges to a centered normal distribution with variance no larger than 1.

In our setting, the key distinction between $H_0^{wb,1}$ and $H_0^{wb,2}$ is a mismatch between two average treatment effects. The conditional sampling distribution is centered at the \textit{focal average effect}:
\begin{equation*}
    \tau(w^S) = \frac{1}{I J_0} \sum_{i=1}^{I} \sum_{j=1}^{J} \{ Y_{ij}(1,0) - Y_{ij}(0,0) \} (1-w^S_j),
\end{equation*}
whereas the weak null hypothesis $H_0^{wb,1}$ in \eqref{eq:weak_null1} concerns the \textit{global average effect}, $\tau$, defined as:
\begin{equation*}
    \tau = \frac{1}{I J} \sum_{i=1}^{I} \sum_{j=1}^{J} \{ Y_{ij}(1,0) - Y_{ij}(0,0) \}.
\end{equation*}
Under $H_0^{wb,2}$, both quantities vanish by construction, so the centering mismatch disappears: $\tau(w^S)=\tau=0$. Under $H_0^{wb,1}$, however, only the global restriction $\tau=0$ is imposed, and $\tau(w^S)$ remains random. In particular, it can be shown that $|\tau(w^S)-\tau|=O_p(1/\sqrt{J})$. Consequently, under $I\asymp J$, the studentized statistic admits the decomposition
\begin{align*}
    \frac{T^B - \tau}{\sqrt{V^B}} = \underbrace{\frac{\tau(w^S) - \tau}{\sqrt{V^B}}}_{A_N} + \underbrace{\frac{T^B - \tau(w^S)}{\sqrt{V^B}}}_{B_N},
\end{align*}
where $T^B := T^{B}(\mathbf{W} \mid \mathbf{Y}, \mathcal{C})$ and $V^B := V^{B}(\mathbf{W} \mid \mathbf{Y}, \mathcal{C})$. As we show in Appendix \ref{app:clt-neyman}, both terms converge to normal limits:
\[
A_N \xrightarrow{d} N(0, \sigma_A^2), \qquad
B_N \xrightarrow{d} N(0, \sigma_B^2),
\]
and hence
\begin{equation*}
    \frac{T^B - \tau}{\sqrt{V^B}} \xrightarrow{d} N(0, \sigma_A^2 + \sigma_B^2).
\end{equation*}
The sampling distribution of $B_N$ is stochastically dominated by $N(0,1)$, whereas $A_N$ reflects additional variation induced by seller randomization. Under $H_0^{wb,2}$, the term $A_N$ vanishes, and the usual stochastic-dominance argument goes through. Under $H_0^{wb,1}$, by contrast, $A_N$ contributes non-negligibly to the sampling distribution, so the total variance $\sigma_A^2+\sigma_B^2$ need not be bounded by 1, the asymptotic variance of the randomization distribution. This creates the possibility of over-rejection under $H_0^{wb,1}$. Figure \ref{fig:t-stats-weak-null} illustrates the limit distributions of the relevant terms in the Neyman-style test statistic under the population-average weak null.

\begin{figure}[h]
    \centering
    \includegraphics[width=0.9\linewidth]{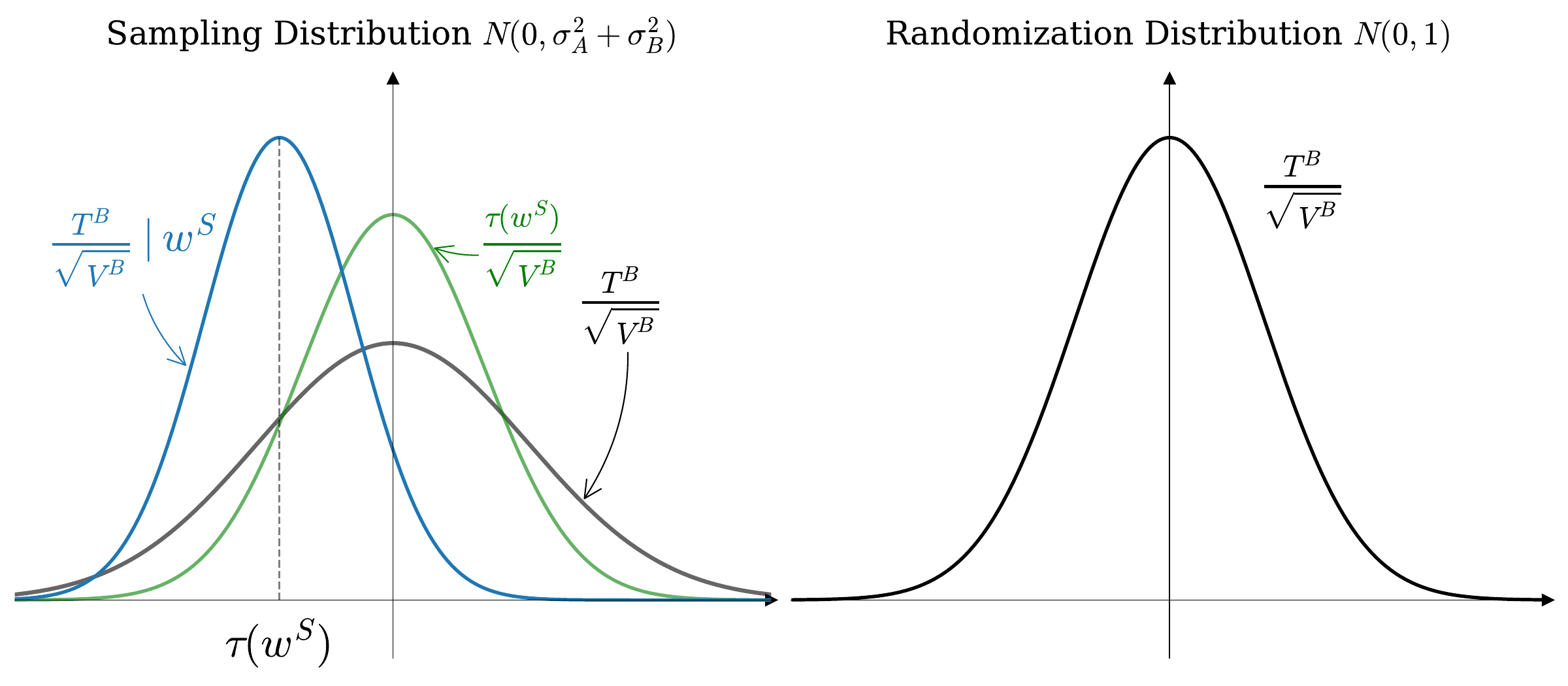}
    \caption{Limit distribution of Neyman-style studentized statistic under $H_0^{wb,1}: \tau = 0$}
    \label{fig:t-stats-weak-null}
\end{figure}

\subsection{Two-way studentization}\label{subsec:weak-null-twoway}
To restore validity under $H_0^{wb,1}$, we augment the standard Neyman variance estimator by an additional term that estimates the sampling variance of the seller-specific mean effect. The resulting variance estimator resembles a design-based version of a two-way (buyer $\times$ seller) cluster-robust variance: the original $V^B$ captures buyer-randomization variation, while the new add-on captures seller-randomization variation.

For each control seller $j$ with $w_j^S=0$, define the within-seller difference-in-means across buyers,
\begin{equation}\label{eq:mu-hat-delta}
    \hat{\mu}^{\Delta}_j(\W\mid \mathbf{Y},\mathcal{C})
    :=
    \frac{1}{I_{1}} \sum_{i=1}^{I} w_i^B \, Y_{i,j}
    -
    \frac{1}{I_{0}} \sum_{i=1}^{I} (1-w_i^B)\, Y_{i,j},
    \qquad (j: w_j^S=0),
\end{equation}
and let $\bar{\hat{\mu}}^{\Delta} := J_0^{-1}\sum_{j:w_j^S=0}\hat{\mu}^{\Delta}_j$. Define the sample variance across control sellers, $s^2_{\hat{\mu}^{\Delta}}(\W\mid \mathbf{Y},\mathcal{C}):= \frac{1}{J_0-1}\sum_{j:w_j^S=0}\left(\hat{\mu}^{\Delta}_j-\bar{\hat{\mu}}^{\Delta}\right)^2$,
and the seller-side add-on
\begin{equation}\label{eq:VS}
    V^{S}(\W\mid \mathbf{Y},\mathcal{C})
    :=
    \left(1-\frac{J_0}{J}\right)\frac{s^2_{\hat{\mu}^{\Delta}}(\W\mid \mathbf{Y},\mathcal{C})}{J_0}.
\end{equation}
Finally define the two-way variance estimator and the corresponding studentized statistic:
\begin{equation}\label{eq:VTW-and-stat}
    \begin{aligned}
        V^{TW}(\W\mid \mathbf{Y},\mathcal{C})
        &:=
        V^{B}(\W\mid \mathbf{Y},\mathcal{C}) + V^{S}(\W\mid \mathbf{Y},\mathcal{C}),\\
        T^{WB,TW}(\W\mid \mathbf{Y},\mathcal{C})
        &:=
        \frac{T^{B}(\W\mid \mathbf{Y},\mathcal{C})}{\sqrt{V^{TW}(\W\mid \mathbf{Y},\mathcal{C})}}.
    \end{aligned}
\end{equation}

To test the population-average weak null $H_0^{wb,1}$, we implement the same permutation procedure as in Procedure~\ref{procedure:weak-spillover}, but replacing $T^{WB}$ with $T^{WB,TW}$ in Steps (i)--(iii). The next theorem states that this two-way studentization restores asymptotic validity under the weak null $H_0^{wb,1}$.

\begin{theorem}[Asymptotic validity under the global weak null via two-way studentization]\label{thm:validity-weak-twoway}
Consider an independent two-sided randomized design under complete randomization where Assumption \ref{ass:local-interference} and Assumptions \ref{assump:bounded-fourth-moments}--\ref{assump:uniform-integrability} hold. Consider the permutation test obtained from Procedure \ref{procedure:weak-spillover} by replacing $T^{WB}$ with $T^{WB,TW}$ in \eqref{eq:VTW-and-stat}. Then for any level $\alpha\in(0,1)$,
\begin{equation*}
    \lim_{I, J \rightarrow \infty} 
    P\left\{
    E\left[ 
    \mathbb{I}\left\{
    T^{WB,TW}(\W\mid \mathbf{Y}^{obs}, \mathcal{C}) \geq T^{WB,TW}(\W^{obs}\mid \mathbf{Y}^{obs}, \mathcal{C})
    \right\} 
    \right] 
    \leq \alpha
    \right\}
    \leq \alpha
\end{equation*}
under the weak null hypothesis $H_0^{wb,1}$ in \eqref{eq:weak_null1}, where the expectation is with respect to $p(\W\mid \mathcal{C})$.
\end{theorem}

\begin{remark}\label{remark:tw-negligible}
The seller-side add-on $V^{S}$ in \eqref{eq:VS} is constructed to estimate the sampling variability of the focal mean effect $\tau(w^S)$ induced by seller randomization. This component is non-negligible under the sampling distribution because $\hat{\mu}_j^\Delta$ estimates a seller-level average treatment effect and the cross-seller variation of these effects potentially remains $O(1)$ as $J\to\infty$.

In contrast, under the randomization distribution induced by permuting buyer labels (holding observed outcomes fixed), each $\hat{\mu}_j^\Delta$ in \eqref{eq:mu-hat-delta} is a difference-in-means computed from a random split of a fixed finite population of buyer outcomes at seller $j$, and therefore typically fluctuates at the $I^{-1/2}$ scale. Consequently, $s^2_{\hat{\mu}^{\Delta}}=O_p(I^{-1})$ under permutations, so $V^{S}=\left(1-\frac{J_0}{J}\right)\frac{s^2_{\hat{\mu}^{\Delta}}}{J_0}=O_p\!\left(\frac{1}{IJ}\right)$, while $V^{B}=O_p(I^{-1})$. Hence $V^{S}/V^{B}=O_p(J^{-1})\to 0$ in the randomization world. This is precisely why $V^{S}$ inflates the sampling denominator—thereby reducing the variance below one—yet is asymptotically negligible for the randomization distribution, which remains $N(0,1)$.
\end{remark}

\section{Simulations}\label{sec:simulation}

\subsection{Sharp null hypotheses}\label{subsec:sim-sharp}

In this subsection, we examine the finite-sample behavior of our randomization tests for both total and spillover effects under the sharp null hypotheses. Suppose that there are $3n$ units in each population, i.e. $I=J=3n$, and the design is a completely randomized two-sided design with $n$ units assigned to treatment group in both populations. We generate the potential outcomes as follows:
\begin{equation*}
    Y_{i j}(w, h) \stackrel{\text { ind }}{\sim} \begin{cases}F_0(\cdot) & \text { if } w=h=0, \\ F_0(\cdot)+F_{B}(\cdot) & \text { if } w=1, h=0, \\ F_0(\cdot)+F_{S}(\cdot) & \text { if } w=0, h=1, \\ F_0(\cdot)+F_{1}(\cdot) & \text { if } w=h=1,\end{cases}
\end{equation*}
where $F_\ell$ are normal distributions such that $F_{\ell}(\cdot)=\mathcal{N}\left(\mu_{\ell}, \sigma_{\ell}^2\right)$ for $\ell \in \{0,B,S,1\}$. Under the sharp null hypothesis, the parameters are set as follows: $\mu_0 = \mu_S = \mu_B = \mu_1 = 0$, $\sigma_0= 0.2$ and $\sigma_B=\sigma_S = \sigma_1 = 0$. 
Under the alternative hypotheses, we set the parameters in a similar way but set the spillover effect to be 0.01 and total effects to be 0.02, i.e. $\mu_B = 0.01, \mu_1 = 0.02$.

\begin{table}[ht!]
\centering
\setlength{\tabcolsep}{3pt}
\begin{adjustbox}{max width=0.9\linewidth,center}
\begin{tabular}{ccccccccccc}
\toprule
 &      &  \multicolumn{3}{c}{Under $H_0$} & & \multicolumn{3}{c}{Under $H_1$} \\ \cmidrule{3-5} \cmidrule{7-9}
Procedure & $n$ & {FRT} & {Neymanian} & {FRT adjusted}  &  &  {FRT} & {Neymanian} & {FRT adjusted}    \\ \midrule
&10 & 5.02  & 0.68  & 4.78   &  & 8.88  & 3.84  & 8.84  \\
Test Buyer
&20 & 4.82  & 3.04  & 5.26   &  & 19.96  & 11.92  & 20.46  \\
Spillover Effects
&30 & 5.00  & 3.66  & 4.90   &  & 42.42  & 28.24  & 41.94  \\
(Procedure \ref{procedure:spillover})
&40 & 4.78  & 2.86  & 5.02   &  & 58.50  & 52.40  & 57.76  \\
&50 & 4.72  & 2.42  & 4.62   &  & 81.42  & 73.66  & 81.26  \\
&100 & 4.74  & 3.28  & 4.84  &  & 100.00  & 100.00  & 100.00  \\
\\
&10 & 4.78  & 1.82  & 5.02   &  & 5.64  & 3.66  & 5.50  \\ 
Test
&20 & 4.84  & 1.24  & 4.94   &  & 10.16  & 23.32  & 10.26  \\
Total Effects
&30 & 4.84  & 0.92  & 4.70   &  & 18.92  & 57.52  & 17.44  \\
(Procedure \ref{procedure:total})
&40 & 4.96  & 1.60  & 5.12   &  & 30.58  & 86.56  & 30.04  \\
&50 & 5.06  & 1.56  & 4.82   &  & 42.58  & 97.90  & 40.34  \\
&100 & 5.24  & 1.28  & 4.92  &  & 95.38  & 100.00  & 93.82  \\
\bottomrule
\end{tabular}
\end{adjustbox}
\caption{Rejection probabilities under sharp null and alternative hypothesis}
\label{table:size-power}
\end{table}

Table \ref{table:size-power} displays the rejection probabilities under the sharp null and alternative hypotheses, computed from 5,000 Monte Carlo replications with the $p$-values approximated by 500 independent permutations of the treatment vector in each replication. ``FRT'' stands for the Fisher Randomization Test procedure based on sharp null hypotheses as in Procedure \ref{procedure:spillover} and \ref{procedure:total}. ``FRT adjusted'' stands for the Neyman-style studentized randomization tests as in Procedure \ref{procedure:weak-spillover}.\footnote{We choose not to present the two-way studentized test statistics in this subsection, because it is developed for the weak null concerning global average effects and only available for buyer spillover effects.} ``Neymanian'' stands for $t$-tests using conservative estimators of variances from \cite{imbens2021}. The block size for Procedure \ref{procedure:total} is set to be $k=\lfloor n/4 \rfloor$, resulting in $|\mathcal{W}|=C(12,4)=495$ and a maximum power approximately equal to 0.955. 

The results show that the rejection probabilities of both ``FRT'' and ``FRT adjusted'' are universally around the nominal 5\% level under the null hypothesis, which verifies the finite-sample exactness of our tests across all designs. Meanwhile, the Neymanian inference method is conservative as expected by \cite{imbens2021}. Under the alternative hypotheses, the rejection probabilities of our tests are higher than the Neymanian method when testing spillover effects, but lower when testing total effects. The power gain by Procedure \ref{procedure:spillover} likely comes from the test's exactness, whereas the power loss from Procedure \ref{procedure:total} is due to the loss in sample size used for calculating the block test statistic.

\begin{table}[ht!]
\centering
\setlength{\tabcolsep}{6pt}
\begin{adjustbox}{max width=0.9\linewidth,center}
\begin{tabular}{cccccccccccc}
\toprule
 &      &  \multicolumn{8}{c}{Block Size $k$}  \\ \cmidrule{3-10} 
Hypothesis & Method & 1 & 2 & 4  & 5  &  10 & 20 & 25 & 50    \\ \midrule
\multirow{2}{*}{$\Htotal$}
& FRT & 5.22  & 5.40  & 5.34  & 5.04  & 4.90  & 4.98  & 4.88  & 0.84  \\
& FRT adjusted & 5.24  & 4.84  & 4.88  & 4.86  & 5.26  & 4.98  & 5.02  & 1.08  \\
\\
\multirow{2}{*}{$H_1^{total}$}
& FRT & 13.20  & 19.56  & 35.90  & 43.46  & 69.40  & 91.06  & 95.10  & 6.98  \\
& FRT adjusted & 13.32  & 20.78  & 35.92  & 41.80  & 69.28  & 89.68  & 93.56  & 13.62  \\
\bottomrule
\end{tabular}
\end{adjustbox}
\caption{Rejection probabilities under different block sizes with $n=100$}
\label{table:size-power-total}
\end{table}

We further investigate robustness by examining the influence of block size on both the size and power of the testing procedures. Table~\ref{table:size-power-total} presents our analysis of rejection probabilities under the sharp null and alternative hypotheses, following the same Monte Carlo setup in Table \ref{table:size-power}. We observe that up to a point the statistical power of our tests increases monotonically with increasing block size. Specifically, when block size is set at \(k=25\), where the randomization space encompasses \(|\mathcal{W}| = C(12,4) = 495\), near-maximum power is achieved. However, a further increase in block size to \(k=50\) results in a diminished randomization space of \(|\mathcal{W}| = C(6,2) = 15\), leading to a notable decline in test power. These results are consistent with the pattern observed in the graphical illustration of the theoretical lower bound in Figure~\ref{fig:power}. 

\subsection{Weak null hypotheses for buyer spillover effects}\label{subsec:sim-weak-null}
In this subsection, we consider two weak null hypotheses for buyer spillover effects introduced in Section~\ref{subsec:weak-null-standard}: $H_0^{wb,1}$, which concerns the population average, and $H_0^{wb,2}$, which concerns the seller-specific average. Appendix~\ref{app:total-effect-weak} presents corresponding results for weak null hypotheses on total effects.

Notably, when treatment effects are i.i.d sampled as in Section \ref{subsec:sim-sharp}, the weak null hypotheses of seller-specific average $H_0^{wb,2}$ are satisfied by the data generating process (DGP). Because the DGP generates treatment effects in a dyadic i.i.d. fashion, the average effects over any sufficiently large subset of the population—including the focal set of untreated sellers—will converge to zero asymptotically. Therefore, we adopt the same DGP from the previous subsection and change the variance parameters as follows: $\sigma_0=0.2, \sigma_B=\sigma_1 = 0.4$ and $\sigma_S = 0$.

In contrast, to examine the population average spillover effect $H_0^{wb, 1}$, we consider a new DGP that satisfies $H_0^{wb,1}$ but violates $H_0^{wb,2}$. Specifically, let $\tilde{Y}_{ij}(0,0)$ and $\tilde{\Delta}_{ij}$ denote the baseline i.i.d.\ components generated as in the previous subsection: $\tilde{Y}_{ij}(0,0)\stackrel{\text{ind}}{\sim}\mathcal{N}(0,0.2^2)$, $\tilde{\Delta}_{ij}\stackrel{\text{ind}}{\sim}\mathcal{N}(0,0.4^2)$. Then, we define $Y_{ij}(1,0)=Y_{ij}(0,0)+\Delta_{ij}$ with 
\[
        Y_{ij}(0,0)=\tilde{Y}_{ij}(0,0)+\alpha_i,
        \qquad
        \Delta_{ij}=\tilde{\Delta}_{ij}+\beta_j,
    \]
where $\alpha_i$ and $\beta_j$ are independent fixed effects, i.e. $\alpha_i \stackrel{\text{ind}}{\sim} \mathcal{N}(0,0.1^2)$ and $\beta_j \stackrel{\text{ind}}{\sim} \mathcal{N}(0,0.4^2)$.\footnote{The remaining potential outcomes $Y_{ij}(0,1)$ and $Y_{ij}(1,1)$ are generated as in the previous subsection and play no role in the buyer-spillover test considered here.} This DGP violates $H_0^{wb,2}$ because seller-specific average treatment effects vary across $j$ through $\beta_j$, so $\tau(w^S)$ need not be close to zero for a realized seller focal set. However, the global average effect $\tau=0$ remains satisfied because both $\tilde{\Delta}_{ij}$ and $\beta_j$ have mean zero.

\begin{table}[ht!]
\centering
\setlength{\tabcolsep}{10pt}
\begin{adjustbox}{max width=0.9\linewidth,center}
\begin{tabular}{ccccc}
\toprule
Setup & $n$ & {FRT} & {FRT adjusted} & {FRT two-way}   \\ \midrule
&10  & 11.24 & 4.52 & 4.38 \\
I.I.D.\ outcomes
&20  & 9.70  & 3.98 & 4.08 \\
and treatment effects
&30  & 8.88  & 3.40 & 3.40 \\
(under $H_0^{wb,2}$)
&40  & 8.46  & 3.74 & 3.28 \\
&50  & 9.20  & 3.56 & 3.42 \\
&100 & 8.60  & 3.62 & 3.44 \\
\\
&10  & 11.04 & 9.52 & 4.04 \\ 
Fixed effects in 
&20  & 16.34 & 15.16 & 4.64 \\
both outcomes and
&30  & 16.20 & 15.10 & 4.76 \\
treatment effects
&40  & 21.62 & 20.24 & 4.90 \\
(under $H_0^{wb,1}$)
&50  & 19.52 & 17.54 & 5.04 \\
&100 & 22.18 & 21.68 & 5.46 \\
\bottomrule
\end{tabular}
\end{adjustbox}
\caption{Rejection probabilities under weak null hypotheses $H_0^{wb,1}$ and $H_0^{wb,2}$.}
\label{table:size-global-weak}
\end{table}

Table~\ref{table:size-global-weak} reports rejection probabilities at the nominal 5\% level based on 5,000 Monte Carlo replications, each using 500 independent permutations. The results in the first half of the table indicate that, under $H_0^{wb,2}$, the adjusted FRT using the Neyman-style studentized statistic successfully maintains the nominal level. So does the two-way adjusted FRT. In contrast, the standard FRT (without studentization) fails to provide valid inference in this setting. This corroborates our theoretical analysis: when seller-specific average effect is zero, the Neyman-style studentization restores the asymptotic validity of the randomization test. In contrast, the results in the second half of the table indicate that, the two-way studentized procedure (``FRT two-way'') controls size well, whereas the standard Neyman-style studentization (``FRT adjusted'') fails due to focal-average randomness induced by seller heterogeneity, as predicted by Theorem \ref{thm:validity-weak}-\ref{thm:validity-weak-twoway}.

\section{Data Application}\label{sec:empirical}

In this section, we illustrate our methodology using a dataset from a randomized field experiment conducted by \citet{Comola2021}, which offered access to formal savings accounts to a random sample of 915 households across 19 villages near Pokhara, Nepal. The treatment is defined at the household level as whether a household was offered a savings account. The primary outcome of interest is the change in a network link indicating whether one household lends to another. These links are measured using survey-based adjacency matrices.\footnote{To be more specific, the network link is constructed from survey questions such as ``Who would you ask for help in case of need?'' and ``Who did you ask for help?''} The matrices are transformed into semi row-standardized versions, $\mathbf{G}^{(t)} = \{G_{i,j}^{(t)}\}$, where each row sums to one for non-isolated households and to zero otherwise and $t \in \{0,1\}$ denotes the baseline and endline periods. The final outcome is defined as $Y_{i,j} = G_{i,j}^{(1)} - G_{i,j}^{(0)}$, representing the change in standardized network links between households $i$ and $j$. 

We follow the empirical question of interest in \citet{Comola2021} and test whether access to a savings account induces financial exchanges between households. In particular, we are interested in exploring heterogeneity of the treatment effect with respect to whether one of the households has experienced a negative shock (death or livestock loss). To formalize this, we construct a binary variable \( X_i \in \{0,1\} \), where \( X_i = 1 \) if household \( i \) experienced a death or livestock shock at baseline and \( X_i = 0 \) otherwise.\footnote{The original dataset includes only household‑pair covariates, but we recover household‑level covariates by aggregating information across all pairwise records involving each household.} This variable acts as a covariate because it is realized before treatment is assigned. The rationale behind this choice is that households experiencing shocks may be in greater financial need and are therefore interpreted as ``buyers'' of informal loans, while those without shocks are potential lenders (``sellers'').

This setting leads naturally to the same potential outcomes structure as in our two-sided experiment framework, with potential outcomes $Y_{i,j}(w_i^B, w_j^S)$, where \( w_i^B, w_j^S \in \{0,1\} \) denote the treatment assignments of households \( i \) and \( j \), respectively. The justification for this structure relies on a different version of the local interference assumption: the outcome for a given household pair depends only on the treatment assignments of the two households in that pair and not on the assignments of any other households in the network. 

The null hypotheses also retain the same structural form as before but carry different empirical interpretations. Specifically, we test three hypotheses corresponding to different treatment comparisons:
\begin{enumerate}
    \item \( H_0^{(1,0)}: Y_{i,j}(0,0) = Y_{i,j}(1,0) \) for all \( i,j \) with \( X_i = 1 \) and \( X_j = 0 \), which tests whether providing access to a savings account to the financially constrained household (buyer) affects its link with an untreated, less constrained household (seller).
    \item \( H_0^{(0,1)}: Y_{i,j}(0,0) = Y_{i,j}(0,1) \) for all \( i,j \) with \( X_i = 1 \) and \( X_j = 0 \), which tests whether treating the less constrained household (seller) alters its connection with an untreated, constrained household (buyer).
    \item \( H_0^{(1,1)}: Y_{i,j}(0,0) = Y_{i,j}(1,1) \) for all \( i,j \) with \( X_i = 1 \) and \( X_j = 0 \), which tests whether jointly treating both households changes their financial connection relative to the case where neither is treated.
\end{enumerate}

Although the original experiment was not designed as a two-sided market intervention, the setting naturally fits our framework. The outcome of interest is dyadic—measured at the household-pair level—so the experimental design can be interpreted as a two-sided randomization problem where the ``buyer'' and ``seller'' sides originate from the same population. In this setting, the outcome $\mathbf{Y}$ is a $915 \times 915$ matrix and the assignment vector is 915‑dimensional, with the buyer‑side and seller‑side assignments coinciding, i.e., $w^B = w^S$. In our analysis, we take a further step to partition households into ``buyers'' and ``sellers'' using baseline information on which households experienced shocks. This not only allows us to more directly apply our two-sided randomization framework but also reflects a meaningful economic distinction that enables analysis of treatment effect heterogeneity.

To test $H_0^{(1,0)}$ and $H_0^{(0,1)}$, we apply the one-sided permutation procedure described in Procedure~\ref{procedure:spillover}. For $\Htotal$, we leverage the fact that the network is censored as survey data record only within-village links. This allows us to use the block-wise permutation procedure in Procedure~\ref{procedure:total}, where each village defines a permutation block.
Table~\ref{table:empirical} reports $p$-values from our randomization-based tests alongside those from a standard two-sample $t$-test.\footnote{We omit the results for ``FRT two-way'' because they are more conservative than those for FRT adjusted'' and are therefore not statistically significant.} We include the $t$-test for comparison, noting that the Neymanian-style methods (\citealp{imbens2021}) are not directly applicable due to the censored nature of the data. In contrast, our approach remains valid by conditioning on the observed household pairs.

\begin{table}[ht!]
\centering
\vspace{20px}
\setlength{\tabcolsep}{15pt}
\begin{adjustbox}{max width=\linewidth,center}
\begin{tabular}{cccc}
\toprule
Hypothesis & FRT & FRT adjusted & Two-sample $t$-test    \\ \midrule
$Y_{i,j}(0,0) = Y_{i,j}(1,0)$ & 0.49 & 0.50 & 0.49\\ 
$Y_{i,j}(0,0) = Y_{i,j}(0,1)$ & 0.42 & 0.41 & 0.41\\ 
$Y_{i,j}(0,0) = Y_{i,j}(1,1)$ & 0.70 & 0.69 & 0.77  \\ 
\bottomrule
\end{tabular}
\end{adjustbox}
\caption{$p$-values on testing the treatment effects on financial network links.}
\label{table:empirical}
\end{table}

We fail to reject all three null hypotheses, indicating that providing savings accounts did not significantly affect households’ risk-sharing relationships. This finding complements \citet{Comola2021}, who document significant effects on network behavior in the full population, which suggests that such effects are not driven by households experiencing shocks. Consistent with prior evidence, the impacts of savings accounts appear highly context-dependent: some studies find limited usage among poor households \citep{Dupas2016}, while others report high take-up but no clear effects on aggregate expenditure, assets, or income \citep{PRINA2015}. Overall, these results suggest that the effectiveness of savings accounts depends critically on how households use them, and that access alone may be insufficient to change financial behavior.

\section{Conclusion}\label{sec:conclusion}
Motivated by recent advances in experimentation within online marketplaces, this paper develops randomization-based inference procedures for two-sided market experiments. Our proposed tests are finite-sample valid under sharp null hypotheses of no treatment effect, and asymptotically valid for weak null hypotheses on average treatment effects. Additionally, we offer practical guidance for test implementation based on power considerations.
Promising directions for future research include extending our framework to accommodate more complex interference structures beyond the cross-product (buyer-seller) interactions examined here. Another extension would incorporate explicit market-clearing mechanisms, such as pricing rules.

\newpage
\spacingset{1}
\bibliography{biblio}

\newpage

\appendix
\section{Further Details}
\subsection{Details for Remark \ref{remark:sample-conditional}}\label{app:details}
Consider the conditional density function $p(\W \mid \mathcal{C})$ with $\W = w^B (w^S)^\top$. With slight abuse of notation, we have
\begin{align*}
    p(\W \mid \mathcal{C}) &= \frac{p(\W, \mathcal{C})}{p(\mathcal{C})} = \frac{p\left(w^{B}, w^{S} = w^{S, obs}  , \sum_{i=1}^I w_i^B = \sum_{i=1}^I w_i^{B, obs}\right) }{p\left(w^{S} = w^{S, obs}, \sum_{i=1}^I w_i = \sum_{i=1}^I w_i^{B, obs}\right)} \\
    &= \frac{p\left(w^{B}, \sum_{i=1}^I w_i^B = \sum_{i=1}^I w_i^{B, obs}\right) p^{S}(w^{S, obs}) }{p^{S}(w^{S, obs}) p\left(\sum_{i=1}^I w_i = \sum_{i=1}^I w_i^{B, obs}\right)} \\
    &= p^B\left(w^{B} \mid  \sum_{i=1}^I w_i^B = \sum_{i=1}^I w_i^{B, obs}\right)~,
\end{align*}
where the final term represents the conditional distribution of $w^B$ on the event $\{\sum_{i=1}^I w_i^B = \sum_{i=1}^I w_i^{B, obs}\}$.
\subsection{Details for Proposition \ref{prop:power}}\label{app:details-power}
\begin{assumption}
    Let $n=|\mathcal{U}|$ and $m = |\mathcal{W}|$. Let the randomization distribution and the null distribution be denoted, respectively, by
    \begin{equation}\label{eqn:prop-power}
        T(\W\mid \mathbf{Y}(\W), \mathcal{C}) \sim \hat F_{1,n,m}, \text{ and } T(\W\mid \mathbf{Y}^{obs}, \mathcal{C}) \sim \hat F_{0,n,m}, \text{ where } \W \sim P(\W \mid \mathcal{C}).
    \end{equation}
    Suppose that for any fixed $n >0$:
    \begin{enumerate}
        \item[(A.1)] There exist continuous cdfs $F_{1,n}$ and $F_{0,n}$ such that $\hat F_{1,n,m}$ and $\hat F_{0,n,m}$ in (25) are the empirical distribution functions over $m$ independent samples from $F_{1,n}$ and $F_{0,n}$, respectively.
        \item[(A.2)] There exists $\sigma_n > 0$, and a continuous cdf $F$, such that $F_{0,n}(t)=F(t/\sigma_n)$, for all $t \in \mathbb{R}$.
        \item[(A.3)] The treatment effect (e.g., spillover contrast) is additive, that is, there exists a fixed $\tau \in \mathbb{R}$ such that $F_{1,n}(t) = F_{0,n}(t-\tau)$, for all $t\in\mathbb{R}$.
    \end{enumerate}
\end{assumption}

\newpage
\section{Proof of Main Results}
\subsection{Proof of Theorem \ref{thm:validity1} and \ref{thm:validity2}}
The proof is structured into two parts. The first part establishes the validity of the randomization test that samples from conditional randomization space by demonstrating that the test statistics are imputable under the conditioning events. The second part confirms the validity of the permutation test by proving the randomization hypotheses.

\subsubsection{Proof for Randomization Test}\label{subsec:proof-randomization-test}
The proof is a direct application of Theorem 2 of \citet{puelz2021}. We first verify that, within the
conditioning event $\mathcal{C}$, the test statistic is \emph{imputable} under the null: for any
$\W' \in \mathcal{W}(\mathcal{C})$,
\[
T(\W'\mid \mathbf{Y}(\W'),\mathcal{C}) = T(\W'\mid \mathbf{Y}(\W),\mathcal{C}).
\]
We first show imputability for the difference-in-means statistics in Procedure~\ref{procedure:spillover} and
Procedure~\ref{procedure:total}, and then for a generic statistic.

\paragraph{Procedure~\ref{procedure:spillover}.}
Recall
\begin{align*}
T(\W\mid \mathbf{Y}(\W), \mathcal{C})
=&\ \frac{1}{n_{1}} \sum_{i=1}^I \sum_{j=1}^J Y_{i,j}(w_i^B,w^{S}_j)\,(1-w^{S}_j)\, w^{B}_i\\
&- \frac{1}{n_{0}}  \sum_{i=1}^I \sum_{j=1}^J Y_{i,j}(w_i^B,w^{S}_j)\,(1-w^{S}_j)\,(1- w^{B}_i).
\end{align*}
Fix any $\W' = w^{B, \prime} ( w^{S, \prime})^\top \in \mathcal{W}(\mathcal{C})$. Under the conditioning event for
Procedure~\ref{procedure:spillover}, the seller-side assignment is fixed, so $w^{S,\prime}=w^S$. Hence,
\begin{align*}
&T(\W'\mid \mathbf{Y}(\W), \mathcal{C})\\
=&\ \frac{1}{n_{1}} \sum_{i=1}^I \sum_{j=1}^J Y_{i,j}(w_i^B,w^{S}_j)\,(1-w^{S}_j)\, w^{B,\prime}_i\\
&- \frac{1}{n_{0}}  \sum_{i=1}^I \sum_{j=1}^J Y_{i,j}(w_i^B,w^{S}_j)\,(1-w^{S}_j)\,(1- w^{B,\prime}_i)\\
=&\ \frac{1}{n_{1}} \sum_{i=1}^I \sum_{j=1}^J Y_{i,j}(w_i^B,0)\,(1-w^{S}_j)\, w^{B,\prime}_i
 - \frac{1}{n_{0}}  \sum_{i=1}^I \sum_{j=1}^J Y_{i,j}(w_i^B,0)\,(1-w^{S}_j)\,(1- w^{B,\prime}_i),
\end{align*}
where the last equality uses that $(1-w_j^S)$ restricts attention to pairs with $w_j^S=0$.

Under $\Hbuyer$, we have $Y_{i,j}(1,0)=Y_{i,j}(0,0)$ for all $i,j$, so for any $z\in\{0,1\}$,
$Y_{i,j}(z,0)=Y_{i,j}(0,0)$. Therefore we may replace $Y_{i,j}(w_i^B,0)$ by $Y_{i,j}(w_i^{B,\prime},0)$:
\begin{align*}
&T(\W'\mid \mathbf{Y}(\W), \mathcal{C})\\
=&\ \frac{1}{n_{1}} \sum_{i=1}^I \sum_{j=1}^J Y_{i,j}(w_i^{B,\prime},0)\,(1-w^{S}_j)\, w^{B,\prime}_i
 - \frac{1}{n_{0}}  \sum_{i=1}^I \sum_{j=1}^J Y_{i,j}(w_i^{B,\prime},0)\,(1-w^{S}_j)\,(1- w^{B,\prime}_i)\\
=&\ T(\W'\mid \mathbf{Y}(\W'), \mathcal{C}),
\end{align*}
which establishes imputability.

\paragraph{Procedure~\ref{procedure:total}.}
Recall
\[
T(\W\mid \mathbf{Y}(\W), \CK)
= \frac{1}{n_{1,k}} \sum_{(i,j)\in\mathcal{U}} Y_{i,j}(w^{B}_i, w^{S}_j)\, w^{S}_j w^{B}_i
- \frac{1}{n_{0,k}} \sum_{(i,j)\in\mathcal{U}} Y_{i,j}(w^{B}_i, w^{S}_j)\, (1-w^{S}_j) (1- w^{B}_i).
\]
Under $\CK$, we have $w_i^B = w_j^S$ for all $(i,j)\in\mathcal{U}$, so
\[
T(\W\mid \mathbf{Y}(\W), \CK)
= \frac{1}{n_{1,k}} \sum_{(i,j)\in\mathcal{U}} Y_{i,j}(w^{B}_i, w^{B}_i)\, w^{B}_i
- \frac{1}{n_{0,k}} \sum_{(i,j)\in\mathcal{U}} Y_{i,j}(w^{B}_i, w^{B}_i)\, (1- w^{B}_i).
\]
Fix $\W' \in \mathcal{W}(\CK)$. Then $\W'$ also satisfies the same block constraint, hence for all $(i,j)\in\mathcal{U}$,
$w_i^{B,\prime}=w_j^{S,\prime}$ and in particular we may write $Y_{i,j}(w_i^{B,\prime},w_i^{B,\prime})$.
By $\Htotal$, $Y_{i,j}(1,1)=Y_{i,j}(0,0)$ for all $i,j$, so $Y_{i,j}(z,z)$ is constant across $z\in\{0,1\}$.
Thus we can replace $Y_{i,j}(w_i^{B},w_i^{B})$ by $Y_{i,j}(w_i^{B,\prime},w_i^{B,\prime})$ inside $T(\W'|\cdot)$:
\begin{align*}
T(\W'\mid \mathbf{Y}(\W), \CK)
=&\ \frac{1}{n_{1,k}} \sum_{(i,j)\in\mathcal{U}} Y_{i,j}(w^{B}_i, w^{B}_i)\, w^{B,\prime}_i
- \frac{1}{n_{0,k}} \sum_{(i,j)\in\mathcal{U}} Y_{i,j}(w^{B}_i, w^{B}_i)\, (1- w^{B,\prime}_i)\\
=&\ \frac{1}{n_{1,k}} \sum_{(i,j)\in\mathcal{U}} Y_{i,j}(w^{B,\prime}_i, w^{B,\prime}_i)\, w^{B,\prime}_i
- \frac{1}{n_{0,k}} \sum_{(i,j)\in\mathcal{U}} Y_{i,j}(w^{B,\prime}_i, w^{B,\prime}_i)\, (1- w^{B,\prime}_i)\\
=&\ T(\W'\mid \mathbf{Y}(\W'), \CK),
\end{align*}
establishing imputability.

\paragraph{Generic statistics.}
The same reasoning applies to any statistic that depends on $(\mathbf{Y}(\W'),\W')$ only through outcomes of pairs
whose exposure is fixed under $\mathcal{C}$.

For Procedure~\ref{procedure:spillover}, $\mathcal{C}$ fixes $w^S$, hence fixes the set
$\{(i,j): w_j^S=0\}$. Under $\Hbuyer$, all potential outcomes $Y_{i,j}(z,0)$ coincide across $z\in\{0,1\}$, so for any
$\W'\in\mathcal{W}(\mathcal{C})$,
\[
\{Y_{i,j}(w_i^B,0)\}_{(i,j):w_j^S=0}
=
\{Y_{i,j}(w_i^{B,\prime},0)\}_{(i,j):w_j^{S,\prime}=0}.
\]
Therefore any such statistic is imputable (note: the null used here is $\Hbuyer$, not $\Htotal$).

For Procedure~\ref{procedure:total}, $\CK$ fixes the constraint $w_i^B=w_j^S$ for $(i,j)\in\mathcal{U}$ and under $\Htotal$
the values $Y_{i,j}(z,z)$ coincide across $z$, so any statistic depending on $\{Y_{i,j}(w_i^B,w_j^S)\}_{(i,j)\in\mathcal{U}}$
through $\{Y_{i,j}(w_i^B,w_i^B)\}_{(i,j)\in\mathcal{U}}$ is imputable by the same substitution.

\paragraph{Correct randomization distribution.}
We show that Procedures~\ref{procedure:spillover} and~\ref{procedure:total}—with Procedure~\ref{procedure:total} replaced by the conditional distribution described in Remark~\ref{remark:sample-conditional} and~\eqref{eq:conditional}—induce randomization distributions consistent with
\[
p(\W \mid \mathcal{C}) \;\propto\; p(\mathcal{C} \mid \W)\, p(\W).
\]
Following the argument in the proof of Theorem~2 of \citet{puelz2021}, it suffices to show that, for every \(\mathbf{w} \in \mathcal{W}\),
\[
p(\mathcal{C} \mid \W = \mathbf{w}) = \mathbb{I}\{\mathbf{w} \in \mathcal{W}\}.
\]
This condition holds for Procedure~\ref{procedure:spillover} by construction.

For Procedure~\ref{procedure:total}, the conditioning mechanism does not favor one focal assignment over another. Indeed, the partitions \(\mathcal{I}^{(k)}(\mathbf{w})\) and \(\mathcal{J}^{(k)}(\mathbf{w})\) are themselves generated at random; equivalently, one may independently shuffle the buyer and seller indices within the treated and control groups and then form blocks of size \(k\) sequentially. This construction is uniform over all assignments compatible with the partition structure: any \(\mathbf{w}\in \mathcal{W}\) induces the event \(\mathcal{C}\) with the same positive probability, whereas any assignment outside \(\mathcal{W}\) induces \(\mathcal{C}\) with probability zero. Therefore,
\[
p(\mathcal{C}\mid \W=\mathbf{w}) \propto \mathbb{I}\{\mathbf{w}\in \mathcal{W}\},
\]
and hence
\[
p(\W\mid \mathcal{C}) \propto p(\mathcal{C}\mid \W)\,p(\W)
\propto \mathbb{I}\{\W\in \mathcal{W}\}\,p(\W),
\]
which is exactly the conditional randomization distribution used in \eqref{eq:conditional}.
\qed

\subsubsection{Proof for Permutation Test}

Apparently, $\W^{(l)}$ belongs to the conditional event, i.e $\W^{(l)} \in \mathcal{W}$ for all $l=1,\dots,L$ under both Procedure \ref{procedure:spillover} and \ref{procedure:total}, meaning that the test statistics are imputable under the null hypothesis and the test procedures are well defined. Then, the validity of those procedures rests upon the \textit{randomization hypothesis} $\W^{(l)} \mid \mathcal{C} \stackrel{d}{=} \W^{obs}\mid \mathcal{C}$ for all $l=1,\dots, L$. To see this, by definition, $\sum_{l=1}^L E[\phi(\W^{(l)}; \mathcal{C}) \mid \mathcal{C}] \leq \alpha L$. By randomization hypothesis, $$\sum_{l=1}^L E[\phi(\W^{(l)}; \mathcal{C}) \mid \mathcal{C} ] = \sum_{l=1}^L E[\phi(\W^{obs}; \mathcal{C}) \mid \mathcal{C}] = L \cdot E[\phi(\W^{obs}; \mathcal{C}) \mid \mathcal{C}] \leq \alpha L~.$$ Therefore, $E[\phi(\W^{obs}; \mathcal{C})\mid \mathcal{C}] \leq \alpha$. 

Then, it suffices to show that randomization hypothesis holds under these two procedures. Let $\pi^B$ denote a permutation operator on the buyer side assignments, and $\pi^{block}$ denote the block-wise permutation. For Procedure \ref{procedure:spillover} (b), by exchangeability of $p^B$, we have $(\pi_B w^B, w^S) \stackrel{d}{=} (w^B, w^S)$. Thus, we have $(\pi_B w^B, w^S) \mid w^S \stackrel{d}{=} (w^B, w^S) \mid w^S$, which implies $\pi_B \W \mid \mathcal{C} \stackrel{d}{=} \W \mid \mathcal{C}$. Similarly, for Procedure \ref{procedure:total} (b), by exchangeability of $p^B$ and $p^S$, we have $(\pi^{block} w^B, \pi^{block} w^S) \stackrel{d}{=} (w^B, w^S)$. Without loss of generality, we assume $I=J$ and $\IK=\JK = \{(1,\dots,k), (k+1,\dots, 2k), \dots\}$, i.e. the focal units consist of $k$ by $k$ diagonal blocks from a square matrix. Thus, we have
\begin{align*}
    &(\pi^{block} w^B, \pi^{block} w^S) \mid \{\pi^{block} w^B = \pi^{block}w^S\} \\
    &=(\pi^{block} w^B, \pi^{block} w^S) \mid \{w^B = w^S\}\\
    &\stackrel{d}{=} (w^B, w^S) \mid \{w^B = w^S\}~,
\end{align*}
which implies $\pi^{block}\W \mid \mathcal{C} \stackrel{d}{=} \W \mid \mathcal{C}$. Note that for both Procedure \ref{procedure:spillover} (b) and \ref{procedure:total} (b), the conditioning event $\mathcal{C}$ has a $\mathcal{W}$ that is more restrictive conditioning set than $\{w^S\}$ (or $\{w^B = w^S\}$). For example, for Procedure \ref{procedure:spillover} (b), $\mathcal{W} = \{v (w^S)^{\top}: v \in  \mathbb{W}^B\} \cap \mathcal{M}$ imposes a restriction of $\mathcal{M}$ (fixed treatment fraction) in addition to $\{w^S\}$. For Procedure \ref{procedure:total} (b), in addition to $\mathcal{M}$, $\mathcal{W}$ requires that $w_i^S = w_j^S$ for all $i,j \in \mathcal{I}_s$ for all $s$ (same for $w^B$), i.e. treatment status should be the same within each partition. That said, those restrictions do not invalidate our argument above, because it is easy to show that $\pi w \in \mathcal{M}$ if and only if $w \in \mathcal{M}$ (same for the additional requirement by Procedure \ref{procedure:total} (b)). Therefore, they are omitted for readability. 

\qed

\subsection{Proof of Section \ref{sec:weak-null}: Randomization Tests for Weak Null Hypotheses}
This section provides the proofs of Theorem \ref{thm:validity-weak} and \ref{thm:validity-weak-twoway} and its supporting lemmas. Throughout this section, we repeatedly use large-sample results for Simple Random Sampling Without Replacement (SRSWOR) to handle the dependencies inherent in finite-population inference.

\subsubsection{Proof of Theorem \ref{thm:validity-weak}}

\begin{proof}
The asymptotic validity of Procedure \ref{procedure:weak-spillover} under $H_0^{wb,2}$ and its potential invalidity under $H_0^{wb,1}$ follow from the limit results established in Lemma \ref{lem:two-term-clt-wb}.

First, consider the null hypothesis $H_0^{wb,2}$ in \eqref{eq:weak_null2}. Under this hypothesis, $\tau(w^S) = \tau = 0$ by construction, causing the first stochastic term $A_N$ to vanish. Lemma \ref{lem:two-term-clt-wb} then implies that the sampling distribution of the studentized statistic converges to a normal distribution with variance $\sigma_B^2 \leq 1$:
\begin{equation*}
    T^{WB} \Rightarrow N(0, \sigma_B^2), \quad \sigma_B^2 \leq 1.
\end{equation*}
Concurrently, the conditional randomization distribution used in Procedure \ref{procedure:weak-spillover}—which is constructed by permuting $w^B$ while holding $w^S$ fixed—converges in probability to a standard normal distribution, $N(0, 1)$, following the general results for studentized statistics under complete randomization \citep{Wu2021,ZHAO2021278}. Because the limit of the randomization distribution has a variance of 1, which is greater than or equal to the sampling variance $\sigma_B^2$, the randomization distribution stochastically dominates the sampling distribution. This ensures that the test is asymptotically conservative, satisfying $\lim \Pr(\text{Reject}) \leq \alpha$.

In contrast, under the global weak null $H_0^{wb,1}$ in \eqref{eq:weak_null1}, Lemma \ref{lem:two-term-clt-wb} shows that the sampling distribution converges to $N(0, \sigma_A^2 + \sigma_B^2)$. As established in the subsequent remarks, the total variance $\sigma_A^2 + \sigma_B^2$ can strictly exceed 1 depending on the values of $S_{\mu^\Delta}^2$ and $\lambda$. In such cases, the sampling distribution is ``wider'' than the $N(0, 1)$ randomization distribution, leading to a systematic over-rejection of the null hypothesis and thus a loss of asymptotic validity.
\end{proof}

\subsubsection{Finite-Population Central Limit Theorems}
We now state the regularity conditions required for our asymptotic theory under two-sided experiments with complete randomization. We consider a sequence of finite populations indexed by $N$, where $I=I_N$ and $J=J_N$. Let $p_1:=I_1/I\to\pi_1\in(0,1)$, $p_0:=I_0/I\to\pi_0\in(0,1)$, $p_2:=J_1/J\to\pi_2\in(0,1)$, and assume $I\to\infty, J\to\infty, I/J\to\lambda\in(0,\infty)$.
\begin{assumption}[Bounded fourth moments]
\label{assump:bounded-fourth-moments}
The potential outcomes satisfy
\[
\sup_N \frac{1}{IJ}\sum_{i=1}^I\sum_{j=1}^J Y_{ij}(z,0)^4<\infty
\quad(z\in\{0,1\}),
\qquad
\sup_N \frac{1}{IJ}\sum_{i=1}^I\sum_{j=1}^J \Delta_{ij}^4<\infty,
\]
where $\Delta_{ij}:=Y_{ij}(1,0)-Y_{ij}(0,0)$.
\end{assumption}

\begin{assumption}[Moment control for seller-averaging error]
\label{assump:moment-control-sellers}
For each $z\in\{0,1\}$, let
\[
\mu_i(z):=\frac{1}{J}\sum_{j=1}^J Y_{ij}(z,0),
\qquad
S_{i,z}^2:=\frac{1}{J-1}\sum_{j=1}^J\big(Y_{ij}(z,0)-\mu_i(z)\big)^2,
\]
and assume
\[
\sup_N \frac{1}{I}\sum_{i=1}^I S_{i,z}^2<\infty \quad (z\in\{0,1\}).
\]
\end{assumption}

\begin{assumption}[Buyer-level nondegeneracy]
\label{assump:buyer-nondegeneracy}
Define $\delta_i:=\mu_i(1)-\mu_i(0)$ and the finite-population variances
\[
S_z^2:=\frac{1}{I-1}\sum_{i=1}^I\big(\mu_i(z)-\bar\mu(z)\big)^2,\quad z\in\{0,1\},
\qquad
S_\delta^2:=\frac{1}{I-1}\sum_{i=1}^I(\delta_i-\bar\delta)^2,
\]
where $\bar\mu(z):=I^{-1}\sum_{i=1}^I \mu_i(z)$ and $\bar\delta:=I^{-1}\sum_{i=1}^I \delta_i$. Assume $S_0^2,S_1^2$ converge to finite limits with $\liminf_N S_0^2>0$ and $\liminf_N S_1^2>0$.
\end{assumption}

\begin{assumption}[Uniform integrability for Lindeberg conditions]
\label{assump:uniform-integrability}
The following uniform integrability conditions hold:
\begin{enumerate}
\item[(i)] For the seller-level treatment effects $\mu_j^\Delta := I^{-1}\sum_{i=1}^I \Delta_{ij}$,
\[
\lim_{M\to\infty} \sup_N \frac{1}{J}\sum_{j=1}^J (\mu_j^\Delta)^2 \mathbf{1}\{(\mu_j^\Delta)^2 > M\} = 0.
\]
\item[(ii)] For each $z\in\{0,1\}$, conditional on any seller assignment $w^S$ with $|S_0(w^S)|=J_0\asymp J$,
\[
\lim_{M\to\infty} \sup_N E\left[\frac{1}{I}\sum_{i=1}^I \big(\bar Y_i^B(z;w^S)\big)^2 \mathbf{1}\big\{\big(\bar Y_i^B(z;w^S)\big)^2 > M\big\}\right] = 0.
\]
\end{enumerate}
\end{assumption}

\begin{remark}
Assumptions~\ref{assump:bounded-fourth-moments}-\ref{assump:uniform-integrability} are natural conditions required for finite-population central limit theorems. The boundedness conditions in Assumptions~\ref{assump:bounded-fourth-moments}, \ref{assump:moment-control-sellers}, and \ref{assump:uniform-integrability} are satisfied if the potential outcomes are bounded, as assumed in \cite{imbens2021}. The nondegeneracy condition in Assumption~\ref{assump:buyer-nondegeneracy} is also a standard assumption for the finite-population CLT \citep{Ding2017}. As we will show later, Assumption~\ref{assump:buyer-nondegeneracy} is essential for ensuring the validity of the studentization method based on our proposed two-way variance estimator.
\end{remark}

The following lemmas establish the necessary central limit theorems for each side of randomization. These results follow from the general finite-population CLT theory developed by \cite{Ding2017}. The key to applying these results is verifying the Lindeberg-type conditions required by \cite{Ding2017}.

\begin{lemma}[Seller-side SRSWOR CLT]
\label{lem:seller-side-clt}
Let $\mu_j^\Delta:=I^{-1}\sum_{i=1}^I \Delta_{ij}$ and define the finite-population variance
\[
S_{\mu^\Delta}^2:=\frac{1}{J-1}\sum_{j=1}^J(\mu_j^\Delta-\tau)^2,
\]
where
\[
\tau:=\frac{1}{IJ}\sum_{i=1}^I\sum_{j=1}^J \big\{Y_{ij}(1,0)-Y_{ij}(0,0)\big\}.
\]
Assume $S_{\mu^\Delta}^2$ converges to a finite positive limit. Under Assumptions~\ref{assump:bounded-fourth-moments}--\ref{assump:uniform-integrability} and the null hypothesis $H_0^{wb}:\tau=0$, we have
\[
\sqrt{J} \tau(w^S) = \sqrt{J}\cdot \frac{1}{J_0}\sum_{j:w_j^S=0}\mu_j^\Delta
\Rightarrow
N\!\left(0,\ \frac{\pi_2}{1-\pi_2}\,S_{\mu^\Delta}^2\right),
\]
where $S_0(w^S):=\{j\in[J]:w_j^S=0\}$ is an SRSWOR subset of $[J]$ of size $J_0$.
\end{lemma}

\begin{proof}
The seller assignment $w^S$ induces an SRSWOR of size $J_0$ from the finite population $\{\mu_j^\Delta\}_{j=1}^J$. Under $H_0^{wb,1}$, the population mean is $\tau=0$.

Define the SRSWOR sample mean $\tau(w^S)=J_0^{-1}\sum_{j:w_j^S=0}\mu_j^\Delta$ and the finite-population variance $S_{\mu^\Delta}^2=(J-1)^{-1}\sum_{j=1}^J (\mu_j^\Delta)^2$ (using $\tau=0$). 

\textbf{Step 1: Verify conditions for \cite{Ding2017} Theorem 1.}

The finite-population CLT for SRSWOR in \cite{Ding2017} (Theorem 1) requires:
\begin{enumerate}
\item \textbf{Sample size growth:} $J_0\to\infty$ and $J-J_0\to\infty$.

This is ensured by $J\to\infty$ and $J_1/J\to\pi_2\in(0,1)$, which implies $J_0=J-J_1\sim (1-\pi_2)J\to\infty$.

\item \textbf{Lindeberg condition:} 
\begin{equation}
\frac{\max_{j\in[J]} (\mu_j^\Delta-\tau)^2}{J\, S_{\mu^\Delta}^2}\to 0.
\label{eq:lindeberg-seller}
\end{equation}
Under $H_0^{wb,1}$, $\tau=0$, so this becomes
\[
\frac{\max_{j\in[J]} (\mu_j^\Delta)^2}{J\, S_{\mu^\Delta}^2}\to 0.
\]
\end{enumerate}

\textbf{Step 2: Derive the Lindeberg condition from Assumption~\ref{assump:uniform-integrability}.}

We need to verify the Lindeberg condition:
\[
\frac{\max_{j\in[J]} (\mu_j^\Delta-\tau)^2}{J\, S_{\mu^\Delta}^2}\to 0.
\]
Under $H_0^{wb,1}$, $\tau=0$, so this becomes
\[
\frac{\max_{j\in[J]} (\mu_j^\Delta)^2}{J\, S_{\mu^\Delta}^2}\to 0.
\]

Since $S_{\mu^\Delta}^2$ converges to a positive limit $\sigma^2>0$, we have $S_{\mu^\Delta}^2 \ge \sigma^2/2$ for sufficiently large $N$. Therefore, it suffices to show
\[
\frac{\max_{j\in[J]} (\mu_j^\Delta)^2}{J} \to 0.
\]

For any $\epsilon > 0$, we have
\[
\frac{1}{J}\max_{j\in[J]} (\mu_j^\Delta)^2
\le
\frac{1}{J}\sum_{j=1}^J (\mu_j^\Delta)^2 \mathbf{1}\{(\mu_j^\Delta)^2 > \epsilon J\}
+
\epsilon.
\]

For the first term, by Assumption~\ref{assump:uniform-integrability}(i) with $M = \epsilon J$:
\[
\sup_N \frac{1}{J}\sum_{j=1}^J (\mu_j^\Delta)^2 \mathbf{1}\{(\mu_j^\Delta)^2 > \epsilon J\} 
\le 
\sup_N \frac{1}{J}\sum_{j=1}^J (\mu_j^\Delta)^2 \mathbf{1}\{(\mu_j^\Delta)^2 > M\} 
\to 0
\]
as $M = \epsilon J \to \infty$ (since $J\to\infty$).

Therefore,
\[
\limsup_{N\to\infty} \frac{1}{J}\max_{j\in[J]} (\mu_j^\Delta)^2 \le \epsilon.
\]
Since $\epsilon>0$ was arbitrary, we conclude
\[
\frac{\max_{j\in[J]} (\mu_j^\Delta)^2}{J} \to 0,
\]
which establishes the Lindeberg condition \eqref{eq:lindeberg-seller}.

\textbf{Step 3: Apply \cite{Ding2017} Theorem 1.}

With both conditions verified, \cite{Ding2017} Theorem 1 yields
\[
\frac{\sqrt{J_0}\,\tau(w^S)}{\sqrt{(1-f_J)S_{\mu^\Delta}^2}}\Rightarrow N(0,1),
\]
where $f_J:=J_0/J$ is the sampling fraction. Since $J_0=J-J_1$ and $J_1/J\to\pi_2$, we have $f_J\to 1-\pi_2$ and thus $1-f_J\to\pi_2$. 

Rearranging gives
\[
\sqrt{J_0}\,\tau(w^S)\Rightarrow N\!\left(0,\ \pi_2 S_{\mu^\Delta}^2\right).
\]
Multiplying both sides by $\sqrt{J/J_0}$ and using $J/J_0\to 1/(1-\pi_2)$ yields
\[
\sqrt{J}\,\tau(w^S)\Rightarrow N\!\left(0,\ \frac{\pi_2}{1-\pi_2}\,S_{\mu^\Delta}^2\right),
\]
completing the proof. 
\end{proof}

For notational brevity, we suppress the superscript ``obs'' and write
\[
(w^B,w^S)=\W^{obs},\qquad \mathbf{Y}^{obs}=\mathbf{Y}(\W^{obs}),
\]
and we abbreviate
\begin{align*}
    T^{B}(w^B,w^S)&:=T^{B}(\W^{obs}\mid \mathbf{Y}^{obs},\mathcal{C})\\
    V^{B}(w^B,w^S)&:=V^{B}(\W^{obs}\mid \mathbf{Y}^{obs},\mathcal{C}) \\
    T^{WB}(w^B,w^S)&:=T^{WB}(\W^{obs}\mid \mathbf{Y}^{obs},\mathcal{C})~.
\end{align*}

\begin{lemma}[Buyer-side CRD CLT conditional on seller assignment]
\label{lem:buyer-side-clt}
Conditional on the seller assignment $w^S$, define the seller-focal buyer-level potential outcomes
\[
\bar Y_i^B(z;w^S)
:=
\frac{1}{J_0}\sum_{j=1}^J Y_{ij}(z,0)\,(1-w_j^S),
\qquad
\bar Y_{\cdot}^B(z;w^S):=\frac{1}{I}\sum_{i=1}^I \bar Y_i^B(z;w^S),
\]
and the seller-focal average treatment effect
\[
\tau(w^S):=\bar Y_{\cdot}^B(1;w^S)-\bar Y_{\cdot}^B(0;w^S).
\]
Let the difference-in-means statistic be
\[
T^{B}(w^B,w^S)
:=
\frac{1}{I_1}\sum_{i=1}^I w_i^B\,\bar Y_i^B(1;w^S)
-
\frac{1}{I_0}\sum_{i=1}^I (1-w_i^B)\,\bar Y_i^B(0;w^S).
\]
Define the conditional finite-population variance
\[
v_I(w^S):=
\frac{S_1^2(w^S)}{I_1}+\frac{S_0^2(w^S)}{I_0}-\frac{S_\tau^2(w^S)}{I},
\]
where
\[
S_z^2(w^S):=\frac{1}{I-1}\sum_{i=1}^I\big(\bar Y_i^B(z;w^S)-\bar Y_\cdot^B(z;w^S)\big)^2,
\]
\[
\delta_i(w^S):=\bar Y_i^B(1;w^S)-\bar Y_i^B(0;w^S),
\qquad
S_\tau^2(w^S):=\frac{1}{I-1}\sum_{i=1}^I\big(\delta_i(w^S)-\tau(w^S)\big)^2.
\]
Under Assumptions~\ref{assump:bounded-fourth-moments}--\ref{assump:uniform-integrability}, conditional on $w^S$, the centered difference-in-means satisfies
\[
\frac{T^{B}(w^B,w^S)-\tau(w^S)}{\sqrt{v_I(w^S)}}
\Rightarrow N(0,1).
\]
\end{lemma}

\begin{proof}
Conditional on $w^S$, the buyer assignment $w^B$ is a two-arm complete randomization with fixed margins $(I_1,I_0)$ applied to the finite population of seller-focal potential outcomes $\{\bar Y_i^B(0;w^S), \bar Y_i^B(1;w^S)\}_{i=1}^I$.

\textbf{Step 1: Verify conditions for \cite{Ding2017}.}

The finite-population CLT for complete randomization in \cite{Ding2017} requires:
\begin{enumerate}
\item \textbf{Sample size growth:} $I_1\to\infty$ and $I_0\to\infty$.

This is ensured by $I\to\infty$ and $I_1/I\to\pi_1\in(0,1)$.

\item \textbf{Lindeberg condition:} For each $z\in\{0,1\}$,
\begin{equation}
\frac{\max_{i\in[I]} \big(\bar Y_i^B(z;w^S)-\bar Y_\cdot^B(z;w^S)\big)^2}{I\, S_z^2(w^S)}\xrightarrow{p} 0.
\label{eq:lindeberg-buyer}
\end{equation}

\item \textbf{Nondegeneracy:} $S_z^2(w^S)$ does not degenerate, which will be established in the proof of the main lemma (Lemma~\ref{lem:two-term-clt-wb}, Step 1(e)) via Assumption~\ref{assump:buyer-nondegeneracy}.
\end{enumerate}

\textbf{Step 2: Derive the Lindeberg condition from Assumption~\ref{assump:uniform-integrability}.}

Fix $z\in\{0,1\}$. We need to verify the Lindeberg condition:
\begin{equation}
\frac{\max_{i\in[I]} \big(\bar Y_i^B(z;w^S)-\bar Y_\cdot^B(z;w^S)\big)^2}{I\, S_z^2(w^S)}\xrightarrow{p} 0.
\label{eq:lindeberg-buyer}
\end{equation}

Since $S_z^2(w^S) \xrightarrow{p} S_z^2 > 0$ (established in Step 1(e) of the main lemma), it suffices to show
\[
\frac{\max_{i\in[I]} \big(\bar Y_i^B(z;w^S)-\bar Y_\cdot^B(z;w^S)\big)^2}{I}\xrightarrow{p} 0.
\]

Using the inequality $\max_i (x_i - \bar x)^2 \le 2\max_i x_i^2 + 2\bar x^2$, we have
\[
\max_{i\in[I]} \big(\bar Y_i^B(z;w^S)-\bar Y_\cdot^B(z;w^S)\big)^2
\le
2\max_{i\in[I]} \big(\bar Y_i^B(z;w^S)\big)^2 + 2\big(\bar Y_\cdot^B(z;w^S)\big)^2.
\]

For the first term, for any $\epsilon > 0$:
\[
\frac{1}{I}\max_{i\in[I]} \big(\bar Y_i^B(z;w^S)\big)^2
\le
\frac{1}{I}\sum_{i=1}^I \big(\bar Y_i^B(z;w^S)\big)^2 \mathbf{1}\big\{\big(\bar Y_i^B(z;w^S)\big)^2 > \epsilon I\big\}
+
\epsilon.
\]

By Assumption~\ref{assump:uniform-integrability}(ii) with $M = \epsilon I$, the first term vanishes as $N\to\infty$ (since $I\to\infty$). Thus
\[
\limsup_{N\to\infty} \frac{1}{I}\max_{i\in[I]} \big(\bar Y_i^B(z;w^S)\big)^2 \le \epsilon.
\]
Since $\epsilon>0$ was arbitrary,
\[
\frac{1}{I}\max_{i\in[I]} \big(\bar Y_i^B(z;w^S)\big)^2 \xrightarrow{p} 0.
\]

For the second term, $\bar Y_\cdot^B(z;w^S)^2 = O_p(1)$ (shown in Step 1(c) of the main lemma), so $\bar Y_\cdot^B(z;w^S)^2/I \to 0$ in probability.

Combining these results establishes \eqref{eq:lindeberg-buyer}.

\textbf{Step 3: Apply \cite{Ding2017}.}

With all conditions verified, \cite{Ding2017} yields the stated result. The normalizing variance $v_I(w^S)$ is precisely the Neyman conservative variance formula for completely randomized experiments, which is the form given in \cite{Ding2017} for two-arm CRD.
\end{proof}

\subsubsection{CLT for Neyman-style Studentized Statistic}\label{app:clt-neyman}

We now present the main asymptotic result for the studentized buyer-spillover statistic under two-sided complete randomization.

\begin{lemma}[Two-term CLT for Neyman-studentized buyer-spillover statistic under two-sided complete randomization]
\label{lem:two-term-clt-wb}
Let there be $I=I_N$ buyers and $J=J_N$ sellers. For each buyer--seller pair $(i,j)$ and each $(b,s)\in\{0,1\}^2$, fix dyad-level potential outcomes $Y_{ij}(b,s)\in\mathbb R$.

\textbf{Two-sided complete randomization with fixed margins.}
Draw the seller assignment $w^S=(w_1^S,\ldots,w_J^S)\in\{0,1\}^J$ uniformly over all vectors with exactly $J_1$ ones and $J_0:=J-J_1$ zeros. Independently draw the buyer assignment $w^B=(w_1^B,\ldots,w_I^B)\in\{0,1\}^I$ uniformly over all vectors with exactly $I_1$ ones and $I_0:=I-I_1$ zeros.

Write $p_1:=I_1/I\to\pi_1\in(0,1)$, $p_0:=I_0/I\to\pi_0\in(0,1)$, $p_2:=J_1/J\to\pi_2\in(0,1)$, and assume
\[
I\to\infty,\qquad J\to\infty,\qquad \frac{I}{J}\to\lambda\in[0,\infty).
\]

\textbf{Neyman variance and studentized statistic.}
Define the Neyman (conservative) variance estimator
\[
V^{B}(w^B,w^S)
:=
\frac{s_1^2(w^B,w^S)}{I_1}+\frac{s_0^2(w^B,w^S)}{I_0},
\]
where $s_1^2(w^B,w^S)=s_1^2$ and $s_0^2(w^B,w^S)=s_0^2$ with $s_1^2, s_0^2$ defined in Procedure \ref{procedure:weak-spillover}, as we use notation $(w^B,w^S)$ to emphasize the dependence on buyer and seller side randomization. 

Define the studentized statistic
\[
T^{WB}(w^B,w^S):=\frac{T^{B}(w^B,w^S)}{\sqrt{V^{B}(w^B,w^S)}}.
\]

\textbf{Two-term decomposition.}
Define
\[
T^{WB}(w^B,w^S)=A_N+B_N,
\qquad
A_N:=\frac{\tau(w^S)}{\sqrt{V^{B}(w^B,w^S)}},
\quad
B_N:=\frac{T^{B}(w^B,w^S)-\tau(w^S)}{\sqrt{V^{B}(w^B,w^S)}}.
\]

\textbf{Asymptotic variances.}
Define the deterministic limit
\[
D:=\frac{S_1^2}{\pi_1}+\frac{S_0^2}{\pi_0},
\]
and
\[
\sigma_A^2:=\lambda\cdot \frac{\pi_2}{1-\pi_2}\cdot\frac{S_{\mu^\Delta}^2}{D},
\qquad
\sigma_B^2:=1-\frac{S_\delta^2}{D}.
\]

Then under Assumptions~\ref{assump:bounded-fourth-moments}--\ref{assump:uniform-integrability} and the null hypothesis $H_0^{wb,1}:\tau=0$,
\[
A_N \Rightarrow N(0,\sigma_A^2),
\qquad
B_N \Rightarrow N(0,\sigma_B^2),
\qquad
(A_N,B_N)\Rightarrow(A,B)\ \text{with }A\perp B,
\]
and therefore
\[
T^{WB}(w^B,w^S)
\Rightarrow
N\!\left(0,\ \sigma_A^2+\sigma_B^2\right).
\]
\end{lemma}

\begin{proof}
\textbf{Step 0: Two-stage structure and $\sigma$-fields.}
Let $\mathcal G_N:=\sigma(w^S)$.
The seller assignment $w^S$ is drawn first; equivalently, $S_0(w^S)$ is an SRSWOR subset of $[J]$ of size $J_0$.
Conditional on $w^S$ (i.e., conditional on $\mathcal G_N$), the buyer assignment $w^B$ is a two-arm complete randomization with fixed margins $(I_1,I_0)$ and is independent of $w^S$.

\textbf{Step 1: Unconditional LLN for the denominator $V^{B}(w^B,w^S)$.}

\textbf{1(a) SRSWOR mean properties for seller focalization.}
Fix a buyer $i$ and $z\in\{0,1\}$ and consider the finite population $\{Y_{ij}(z,0)\}_{j=1}^J$.
Because $S_0(w^S)=\{j:w_j^S=0\}$ is an SRSWOR subset of size $J_0$, the focal mean
\[
\bar Y_i^B(z;w^S)=\frac{1}{J_0}\sum_{j:w_j^S=0}Y_{ij}(z,0)
\]
is an SRSWOR mean. Therefore
\[
\E\big[\bar Y_i^B(z;w^S)\big]=\mu_i(z),
\]
and the SRSWOR variance formula gives
\begin{equation}
\Var\big(\bar Y_i^B(z;w^S)\big)
=
\Big(1-\frac{J_0}{J}\Big)\frac{S_{i,z}^2}{J_0},
\label{eq:var-srswor}
\end{equation}
where $S_{i,z}^2$ is the finite-population variance over sellers defined in Assumption~\ref{assump:moment-control-sellers}.
Since the SRSWOR mean is unbiased,
\[
\E\!\left[\big(\bar Y_i^B(z;w^S)-\mu_i(z)\big)^2\right]
=
\Var\big(\bar Y_i^B(z;w^S)\big).
\]
Combining with \eqref{eq:var-srswor},
\begin{equation}
\E\!\left[\big(\bar Y_i^B(z;w^S)-\mu_i(z)\big)^2\right]
=
\Big(1-\frac{J_0}{J}\Big)\frac{S_{i,z}^2}{J_0}.
\label{eq:l2i}
\end{equation}

\textbf{1(b) Average $L^2$ seller-averaging error vanishes.}
Average \eqref{eq:l2i} over buyers $i=1,\dots,I$:
\begin{align}
\frac{1}{I}\sum_{i=1}^I
\E\!\left[\big(\bar Y_i^B(z;w^S)-\mu_i(z)\big)^2\right]
&=
\Big(1-\frac{J_0}{J}\Big)\frac{1}{J_0}\cdot\frac{1}{I}\sum_{i=1}^I S_{i,z}^2.
\label{eq:avg-l2}
\end{align}
By Assumption~\ref{assump:moment-control-sellers}, $\sup_N I^{-1}\sum_i S_{i,z}^2<\infty$, and $J_0\asymp J$, so the right-hand side is $O(J^{-1})\to 0$.
Define the nonnegative random variable
\[
\Delta_{z,N}:=\frac{1}{I}\sum_{i=1}^I\big(\bar Y_i^B(z;w^S)-\mu_i(z)\big)^2.
\]
Then $\E[\Delta_{z,N}]\to 0$ by \eqref{eq:avg-l2}.
By Markov's inequality, for any $\epsilon>0$,
\[
\Pr(\Delta_{z,N}>\epsilon)\le \frac{\E[\Delta_{z,N}]}{\epsilon}\to 0,
\]
hence
\begin{equation}
\Delta_{z,N}\xrightarrow{p}0\qquad (z\in\{0,1\}).
\label{eq:Delta}
\end{equation}

\textbf{1(c) Stability of buyer-level means.}
By Cauchy--Schwarz,
\[
\left|
\frac{1}{I}\sum_{i=1}^I\big(\bar Y_i^B(z;w^S)-\mu_i(z)\big)
\right|^2
\le
\frac{1}{I}\sum_{i=1}^I\big(\bar Y_i^B(z;w^S)-\mu_i(z)\big)^2
=\Delta_{z,N}.
\]
Thus \eqref{eq:Delta} implies
\[
\bar Y_{\cdot}^B(z;w^S)-\bar\mu(z)\xrightarrow{p}0.
\]

\textbf{1(d) Stability of second moments.}
Write
\[
\bar Y_i^B(z;w^S)^2-\mu_i(z)^2
=
\big(\bar Y_i^B(z;w^S)-\mu_i(z)\big)\big(\bar Y_i^B(z;w^S)+\mu_i(z)\big).
\]
Averaging over $i$ and applying Cauchy--Schwarz gives
\begin{align}
\left|
\frac{1}{I}\sum_{i=1}^I\bar Y_i^B(z;w^S)^2
-
\frac{1}{I}\sum_{i=1}^I\mu_i(z)^2
\right|
&\le
\Big(\Delta_{z,N}\Big)^{1/2}
\left(
\frac{1}{I}\sum_{i=1}^I\big(\bar Y_i^B(z;w^S)+\mu_i(z)\big)^2
\right)^{1/2}.
\label{eq:2nd-mom-bound}
\end{align}
The first factor $\Delta_{z,N}^{1/2}\to_p 0$ by \eqref{eq:Delta}.
It remains to show the second factor is $O_p(1)$.

Using $(x+y)^2\le 2x^2+2y^2$,
\[
\frac{1}{I}\sum_{i=1}^I(\bar Y_i^B(z;w^S)+\mu_i(z))^2
\le
\frac{2}{I}\sum_{i=1}^I\bar Y_i^B(z;w^S)^2+\frac{2}{I}\sum_{i=1}^I\mu_i(z)^2.
\]
By Jensen,
\[
\mu_i(z)^4=\Big(\frac{1}{J}\sum_{j=1}^J Y_{ij}(z,0)\Big)^4 \le \frac{1}{J}\sum_{j=1}^J Y_{ij}(z,0)^4,
\]
and similarly
\[
\bar Y_i^B(z;w^S)^4=\Big(\frac{1}{J_0}\sum_{j:w_j^S=0} Y_{ij}(z,0)\Big)^4 \le \frac{1}{J_0}\sum_{j:w_j^S=0} Y_{ij}(z,0)^4.
\]
Averaging over $i$ and using Assumption~\ref{assump:bounded-fourth-moments} implies
\[
\sup_N \frac{1}{I}\sum_{i=1}^I \mu_i(z)^4<\infty,
\qquad
\sup_N \E\!\left[\frac{1}{I}\sum_{i=1}^I \bar Y_i^B(z;w^S)^4\right]<\infty.
\]
Hence, by Markov's inequality (tightness), both $(I^{-1}\sum_i\mu_i(z)^2)$ and $(I^{-1}\sum_i \bar Y_i^B(z;w^S)^2)$ are $O_p(1)$, which shows the second factor in \eqref{eq:2nd-mom-bound} is $O_p(1)$.
Therefore the left-hand side of \eqref{eq:2nd-mom-bound} converges to $0$ in probability:
\begin{equation}
\frac{1}{I}\sum_{i=1}^I\bar Y_i^B(z;w^S)^2
-
\frac{1}{I}\sum_{i=1}^I\mu_i(z)^2
\xrightarrow{p}0.
\label{eq:2nd-mom-stable}
\end{equation}

\textbf{1(e) Variance stability: $S_z^2(w^S)\to_p S_z^2$.}
Use the identity, valid for any $x_1,\dots,x_I$,
\[
\frac{1}{I-1}\sum_{i=1}^I(x_i-\bar x)^2=\frac{I}{I-1}\left(\frac{1}{I}\sum_{i=1}^I x_i^2-\bar x^2\right).
\]
Apply it to $x_i=\bar Y_i^B(z;w^S)$ and to $x_i=\mu_i(z)$ to get
\begin{align*}
S_z^2(w^S)-S_z^2
&=
\frac{I}{I-1}\Bigg[
\left(\frac{1}{I}\sum_{i=1}^I \bar Y_i^B(z;w^S)^2 - \bar Y_\cdot^B(z;w^S)^2\right)
-
\left(\frac{1}{I}\sum_{i=1}^I \mu_i(z)^2 - \bar\mu(z)^2\right)
\Bigg]
\\
&=
\frac{I}{I-1}\Bigg[
\left(\frac{1}{I}\sum_{i=1}^I \bar Y_i^B(z;w^S)^2-\frac{1}{I}\sum_{i=1}^I \mu_i(z)^2\right)
-
\left(\bar Y_\cdot^B(z;w^S)^2-\bar\mu(z)^2\right)
\Bigg].
\end{align*}
The first bracket converges to $0$ in probability by \eqref{eq:2nd-mom-stable}.
For the second bracket,
\[
\bar Y_\cdot^B(z;w^S)^2-\bar\mu(z)^2
=
\big(\bar Y_\cdot^B(z;w^S)-\bar\mu(z)\big)\big(\bar Y_\cdot^B(z;w^S)+\bar\mu(z)\big).
\]
The first factor $\to_p 0$ by Step 1(c); the second factor is $O_p(1)$ by the same tightness argument as in Step 1(d).
Hence $\bar Y_\cdot^B(z;w^S)^2-\bar\mu(z)^2\to_p 0$, and since $I/(I-1)\to 1$ we conclude
\begin{equation}
S_z^2(w^S)\xrightarrow{p}S_z^2\qquad (z\in\{0,1\}).
\label{eq:Sz-stable}
\end{equation}

\textbf{1(f) Denominator convergence.}
Define the finite-population limit
\[
D:=\frac{S_1^2}{\pi_1}+\frac{S_0^2}{\pi_0}.
\]
Since $p_z=I_z/I\to\pi_z$ and \eqref{eq:Sz-stable} holds,
\[
I\left(\frac{S_1^2(w^S)}{I_1}+\frac{S_0^2(w^S)}{I_0}\right)
=
\frac{S_1^2(w^S)}{p_1}+\frac{S_0^2(w^S)}{p_0}
\xrightarrow{p}
\frac{S_1^2}{\pi_1}+\frac{S_0^2}{\pi_0}
=D.
\]
Finally, the sample-variance version $V^{B}(w^B,w^S)$ differs from $S_1^2(w^S)/I_1+S_0^2(w^S)/I_0$ only by replacing finite-population variances with within-arm sample variances; under complete randomization, these are consistent for their finite-population counterparts, and in particular
\begin{equation}
I\,V^{B}(w^B,w^S)\xrightarrow{p}D.
\label{eq:denom-conv}
\end{equation}

\textbf{Step 2: Limit law for $A_N=\tau(w^S)/\sqrt{V^{B}}$.}

\textbf{2(a) Rewrite $\tau(w^S)$ as an SRSWOR mean over sellers.}
Let $\Delta_{ij}:=Y_{ij}(1,0)-Y_{ij}(0,0)$ and define seller-level averages
\[
\mu_j^\Delta:=\frac{1}{I}\sum_{i=1}^I \Delta_{ij}.
\]
Then
\[
\tau(w^S)
=
\frac{1}{IJ_0}\sum_{i=1}^I\sum_{j=1}^J \Delta_{ij}(1-w_j^S)
=
\frac{1}{J_0}\sum_{j:w_j^S=0}\mu_j^\Delta.
\]
Also,
\[
\frac{1}{J}\sum_{j=1}^J\mu_j^\Delta
=
\frac{1}{IJ}\sum_{i=1}^I\sum_{j=1}^J \Delta_{ij}
=\tau.
\]
Under $H_0^{wb}:\tau=0$, the seller-population mean of $\{\mu_j^\Delta\}$ equals $0$.

\textbf{2(b) Apply the seller-side SRSWOR CLT and rescale.}
By Lemma~\ref{lem:seller-side-clt},
\[
\sqrt{J}\,\tau(w^S)
=
\sqrt{J}\cdot \frac{1}{J_0}\sum_{j:w_j^S=0}\mu_j^\Delta
\Rightarrow
N\!\left(0,\ \frac{\pi_2}{1-\pi_2}\,S_{\mu^\Delta}^2\right).
\]
Now write
\[
A_N
=
\frac{\sqrt{I}\,\tau(w^S)}{\sqrt{I\,V^{B}}}
=
\sqrt{\frac{I}{J}}\cdot \frac{\sqrt{J}\,\tau(w^S)}{\sqrt{I\,V^{B}}}.
\]
Using $I/J\to\lambda$ and $I\,V^B\to_p D$ from \eqref{eq:denom-conv}, Slutsky's theorem yields
\[
A_N\Rightarrow
N\!\left(0,\ \lambda\cdot \frac{\pi_2}{1-\pi_2}\cdot\frac{S_{\mu^\Delta}^2}{D}\right)
=:N(0,\sigma_A^2).
\]

\textbf{Step 3: Limit law for $B_N=(T^{B}-\tau(w^S))/\sqrt{V^{B}}$.}

\textbf{3(a) Factor $B_N$ into a CLT part and a variance-ratio part.}
Define
\[
Z_N(w^S):=\frac{T^{B}(w^B,w^S)-\tau(w^S)}{\sqrt{v_I(w^S)}},
\qquad
R_N(w^S):=\sqrt{\frac{v_I(w^S)}{V^{B}(w^B,w^S)}}.
\]
Then
\begin{equation}
B_N=Z_N(w^S)\cdot R_N(w^S).
\label{eq:BN-factor}
\end{equation}
By Lemma~\ref{lem:buyer-side-clt}, conditional on $w^S$,
\begin{equation}
Z_N(w^S)\Rightarrow N(0,1).
\label{eq:ZN-clt}
\end{equation}

\textbf{3(b) Show $R_N(w^S)\to_p \sqrt{1-S_\delta^2/D}$.}
By definition,
\[
v_I(w^S)
=
\frac{S_1^2(w^S)}{I_1}+\frac{S_0^2(w^S)}{I_0}-\frac{S_\tau^2(w^S)}{I}.
\]
Also
\[
\frac{S_1^2(w^S)}{I_1}+\frac{S_0^2(w^S)}{I_0}
=
\frac{1}{I}\left(\frac{S_1^2(w^S)}{p_1}+\frac{S_0^2(w^S)}{p_0}\right).
\]
From Step 1(e), $S_z^2(w^S)\to_p S_z^2$ and hence
\[
\frac{S_1^2(w^S)}{p_1}+\frac{S_0^2(w^S)}{p_0}\to_p D,
\quad\text{so}\quad
\frac{S_1^2(w^S)}{I_1}+\frac{S_0^2(w^S)}{I_0}
=
\frac{D+o_p(1)}{I}.
\]
It remains to control $S_\tau^2(w^S)$.

Define the focal buyer-level effects
\[
\delta_i(w^S):=\bar Y_i^B(1;w^S)-\bar Y_i^B(0;w^S),
\]
and the full-seller buyer-level effects
\[
\delta_i:=\mu_i(1)-\mu_i(0).
\]
Then
\[
\delta_i(w^S)-\delta_i
=
\big(\bar Y_i^B(1;w^S)-\mu_i(1)\big)
-
\big(\bar Y_i^B(0;w^S)-\mu_i(0)\big).
\]
Using $(x-y)^2\le 2x^2+2y^2$ and \eqref{eq:Delta} for $z=1,0$,
\[
\frac{1}{I}\sum_{i=1}^I(\delta_i(w^S)-\delta_i)^2
\le
\frac{2}{I}\sum_{i=1}^I(\bar Y_i^B(1;w^S)-\mu_i(1))^2
+
\frac{2}{I}\sum_{i=1}^I(\bar Y_i^B(0;w^S)-\mu_i(0))^2
\xrightarrow{p}0.
\]
Denote $\bar\delta(w^S):=I^{-1}\sum_i\delta_i(w^S)=\tau(w^S)$ and $\bar\delta:=I^{-1}\sum_i\delta_i=\tau$.
Applying the same variance-stability argument as in Step 1(e) (with $x_i=\delta_i(w^S)$ and $x_i=\delta_i$) yields
\begin{equation}
S_\tau^2(w^S)
=
\frac{1}{I-1}\sum_{i=1}^I\big(\delta_i(w^S)-\bar\delta(w^S)\big)^2
\xrightarrow{p}
\frac{1}{I-1}\sum_{i=1}^I(\delta_i-\bar\delta)^2
=:S_\delta^2.
\label{eq:Stau-to-Sdelta}
\end{equation}

Now use the identity
\[
v_I(w^S)
=
\left(\frac{S_1^2(w^S)}{I_1}+\frac{S_0^2(w^S)}{I_0}\right)-\frac{S_\tau^2(w^S)}{I}.
\]
Multiply both sides by $I$:
\[
I\,v_I(w^S)
=
\frac{S_1^2(w^S)}{p_1}+\frac{S_0^2(w^S)}{p_0}-S_\tau^2(w^S)
\xrightarrow{p}
D-S_\delta^2,
\]
using $S_z^2(w^S)\to_p S_z^2$ and \eqref{eq:Stau-to-Sdelta}.
Also from \eqref{eq:denom-conv}, $I\,V^{B}\to_p D$.
Therefore,
\[
R_N(w^S)^2
=
\frac{v_I(w^S)}{V^{B}}
=
\frac{I\,v_I(w^S)}{I\,V^{B}}
\xrightarrow{p}
\frac{D-S_\delta^2}{D}
=
1-\frac{S_\delta^2}{D}.
\]
Taking square roots gives
\begin{equation}
R_N(w^S)\xrightarrow{p}\sqrt{1-\frac{S_\delta^2}{D}}.
\label{eq:R-limit}
\end{equation}

\textbf{3(c) Conclude $B_N\Rightarrow N(0,\sigma_B^2)$.}
From \eqref{eq:BN-factor}, conditional on $w^S$, $Z_N(w^S)\Rightarrow N(0,1)$ by \eqref{eq:ZN-clt}, and $R_N(w^S)\to_p \sqrt{1-S_\delta^2/D}$ by \eqref{eq:R-limit}. A conditional Slutsky theorem implies that conditional on $w^S$,
\[
B_N \Rightarrow N\!\left(0,\ 1-\frac{S_\delta^2}{D}\right).
\]
Since the limiting variance is deterministic, this implies the unconditional convergence
\[
B_N \Rightarrow N\!\left(0,\ 1-\frac{S_\delta^2}{D}\right)=:N(0,\sigma_B^2).
\]

\textbf{Step 4: Joint convergence and asymptotic independence.}
$A_N$ is $\mathcal G_N$-measurable (it is a function of $w^S$ and $V^B$, whose limit depends only on $w^S$ through concentrating quantities), while the randomness generating $B_N$ comes only from the buyer randomization $w^B$, which is independent of $w^S$.

For $t_1,t_2\in\mathbb R$, consider the joint characteristic function
\[
\phi_N(t_1,t_2):=\E\Big[\exp\{i t_1 A_N+i t_2 B_N\}\Big].
\]
Using the tower property and $\mathcal G_N$-measurability of $A_N$,
\[
\phi_N(t_1,t_2)
=
\E\Big[
\exp\{i t_1 A_N\}\,
\E\big[\exp\{i t_2 B_N\}\mid \mathcal G_N\big]
\Big].
\]
From Step 3(c), conditional on $\mathcal G_N$ the distribution of $B_N$ converges to $N(0,\sigma_B^2)$ with deterministic variance, so
\[
\E\big[\exp\{i t_2 B_N\}\mid \mathcal G_N\big]
\xrightarrow{p}
\exp\!\left(-\tfrac{1}{2}t_2^2\sigma_B^2\right).
\]
Also $\big|\E[\exp\{i t_2 B_N\}\mid \mathcal G_N]\big|\le 1$, so bounded convergence yields
\[
\phi_N(t_1,t_2)
-
\E\big[\exp\{i t_1 A_N\}\big]\exp\!\left(-\tfrac{1}{2}t_2^2\sigma_B^2\right)
\to 0.
\]
From Step 2, $A_N\Rightarrow N(0,\sigma_A^2)$, hence
\[
\E\big[\exp\{i t_1 A_N\}\big]\to \exp\!\left(-\tfrac{1}{2}t_1^2\sigma_A^2\right).
\]
Therefore
\[
\phi_N(t_1,t_2)\to
\exp\!\left(-\tfrac{1}{2}t_1^2\sigma_A^2\right)
\exp\!\left(-\tfrac{1}{2}t_2^2\sigma_B^2\right),
\]
the characteristic function of independent normals. Hence
\[
(A_N,B_N)\Rightarrow(A,B)\quad\text{with }A\perp B,\ A\sim N(0,\sigma_A^2),\ B\sim N(0,\sigma_B^2).
\]

\textbf{Step 5: Sum.}
Since $T^{WB}=A_N+B_N$ and the limit is the sum of independent normals,
\[
T^{WB}(w^B,w^S)
\Rightarrow
N\!\left(0,\ \sigma_A^2+\sigma_B^2\right)
=
N\!\left(0,\ 1-\frac{S_\delta^2}{D}
+\lambda\cdot \frac{\pi_2}{1-\pi_2}\cdot \frac{S_{\mu^\Delta}^2}{D}\right).
\]
\end{proof}

\subsubsection{Proof of Theorem~\ref{thm:validity-weak-twoway}}\label{app:proof-validity-weak-twoway}

\begin{proof}
Fix a sequence of finite populations satisfying Assumption~\ref{ass:local-interference} and Assumptions~\ref{assump:bounded-fourth-moments}--\ref{assump:uniform-integrability}. Throughout, the conditioning event $\mathcal{C}=(\mathcal{U},\mathcal{W})$ is as in Procedure~\ref{procedure:weak-spillover}, and the randomization distribution is $p(\W\mid \mathcal{C})$ induced by permuting buyer labels (holding $w^{S}$ fixed at $w^{S,obs}$).

Define the two-way studentized statistic
\[
T^{WB,TW}(\W\mid \mathbf{Y},\mathcal{C})
=
\frac{T^{B}(\W\mid \mathbf{Y},\mathcal{C})}{\sqrt{V^{TW}(\W\mid \mathbf{Y},\mathcal{C})}}
\quad\text{with}\quad
V^{TW}=V^{B}+V^{S},
\]
where $V^{B}$ is the standard Neyman variance in Procedure~\ref{procedure:weak-spillover} and $V^{S}$ is defined in \eqref{eq:VS}.

Let the (idealized) one-sided randomization $p$-value be
\[
p_N
:=
E\!\left[
\mathbb{I}\left\{
T^{WB,TW}(\W\mid \mathbf{Y}^{obs},\mathcal{C})
\ge
T^{WB,TW}(\W^{obs}\mid \mathbf{Y}^{obs},\mathcal{C})
\right\}
\right],
\]
where the expectation is with respect to $p(\W\mid \mathcal{C})$. The theorem is proved by establishing two ingredients:
\begin{enumerate}
\item[(I)] conditional on $(\mathbf{Y}^{obs},\mathcal{C})$, the randomization distribution of $T^{WB,TW}(\W\mid \mathbf{Y}^{obs},\mathcal{C})$ converges to $N(0,1)$;
\item[(II)] under $H_0^{wb,1}$, the sampling distribution of $T^{WB,TW}(\W^{obs}\mid \mathbf{Y}^{obs},\mathcal{C})$ converges to a centered normal with variance weakly smaller than $1$.
\end{enumerate}
Given (I)--(II), asymptotic validity follows by the standard stochastic domination argument.

\textbf{Step 1: Randomization distribution.}
Condition on $(\mathbf{Y}^{obs},\mathcal{C})$ and write $P_\pi(\cdot)$ and $E_\pi[\cdot]$ for probability and expectation under the permutation distribution $p(\W\mid\mathcal{C})$ (i.e., permuting $w^B$ with fixed margins $(I_1,I_0)$ while holding $w^{S,obs}$ fixed). All $O_p(\cdot)$ and $o_p(\cdot)$ statements in this step are with respect to $P_\pi$, conditional on $(\mathbf{Y}^{obs},\mathcal{C})$.

Treat the buyer-level focal means $\{\bar Y_i^B\}_{i=1}^I$ (computed using the fixed $w^{S,obs}$) as a fixed finite population. Under $P_\pi$, $w^B$ is a complete randomization over buyers with fixed margins $(I_1,I_0)$. Hence, by the finite-population CLT for complete randomization and standard studentization arguments (e.g., \citealp{Ding2017}), the studentized statistic
\[
T^{WB}(\W\mid \mathbf{Y}^{obs},\mathcal{C})
=
\frac{T^{B}(\W\mid \mathbf{Y}^{obs},\mathcal{C})}{\sqrt{V^{B}(\W\mid \mathbf{Y}^{obs},\mathcal{C})}}
\]
has a conditional randomization distribution converging to $N(0,1)$, in the sense that
\[
\sup_{t\in\mathbb{R}}
\left|
P_\pi\!\left(
T^{WB}(\W\mid \mathbf{Y}^{obs},\mathcal{C})\le t 
\right)
-
\Phi(t)
\right|
\;\to\; 0
\quad\text{in probability (w.r.t.\ the sampling distribution).}
\]

It remains to show that replacing $V^{B}$ by $V^{TW}=V^{B}+V^{S}$ does not change the randomization limit. We do so by establishing the relative rates
\[
V^{B}(\W\mid \mathbf{Y}^{obs},\mathcal{C}) = O_p(I^{-1})
\qquad\text{and}\qquad
V^{S}(\W\mid \mathbf{Y}^{obs},\mathcal{C}) = O_p\!\left(\frac{1}{IJ}\right),
\]
which implies $V^{S}/V^{B}=O_p(J^{-1})\to 0$ under $P_\pi$.

\textbf{Step 1(a): Rate of $V^{S}$ under $P_\pi$.}
Fix a control seller $j$ with $w^{S,obs}_j=0$ and define the fixed finite population across buyers
\[
y_{ij}:=Y^{obs}_{ij}\qquad(i=1,\dots,I),
\qquad
\bar y_{\cdot j}:=\frac{1}{I}\sum_{i=1}^I y_{ij},
\qquad
S_{j,obs}^2:=\frac{1}{I-1}\sum_{i=1}^I (y_{ij}-\bar y_{\cdot j})^2.
\]
Under $P_\pi$, the treated buyers $\{i:w_i^B=1\}$ form an SRSWOR subset of $[I]$ of size $I_1$ and the control buyers form the complement of size $I_0$. Writing
\[
\bar y_{1j}(\W):=\frac{1}{I_1}\sum_{i:w_i^B=1}y_{ij},
\qquad
\bar y_{0j}(\W):=\frac{1}{I_0}\sum_{i:w_i^B=0}y_{ij},
\]
we can rewrite \eqref{eq:mu-hat-delta} as
\[
\hat{\mu}_j^\Delta(\W\mid \mathbf{Y}^{obs},\mathcal{C})=\bar y_{1j}(\W)-\bar y_{0j}(\W).
\]
Because both $\bar y_{1j}(\W)$ and $\bar y_{0j}(\W)$ are SRSWOR means from the same finite population, their expectations equal $\bar y_{\cdot j}$. Moreover, the standard SRSWOR variance identities give
\[
\Var_\pi\!\big(\bar y_{1j}(\W)\big)=\left(1-\frac{I_1}{I}\right)\frac{S_{j,obs}^2}{I_1},
\qquad
\Var_\pi\!\big(\bar y_{0j}(\W)\big)=\left(1-\frac{I_0}{I}\right)\frac{S_{j,obs}^2}{I_0},
\]
and
\[
\cov_\pi\!\big(\bar y_{1j}(\W),\bar y_{0j}(\W)\big)=-\frac{S_{j,obs}^2}{I}.
\]
Combining these three identities yields the exact variance formula
\begin{equation}\label{eq:var-muhat-perm}
\Var_\pi\!\left(\hat{\mu}_j^\Delta(\W\mid \mathbf{Y}^{obs},\mathcal{C})\right)
=
\left(\frac{1}{I_1}+\frac{1}{I_0}\right)S_{j,obs}^2.
\end{equation}
In particular,
\[
E_\pi\!\left[\big(\hat{\mu}_j^\Delta(\W\mid \mathbf{Y}^{obs},\mathcal{C})\big)^2\right]
=
\Var_\pi\!\left(\hat{\mu}_j^\Delta(\W\mid \mathbf{Y}^{obs},\mathcal{C})\right)
=
\left(\frac{1}{I_1}+\frac{1}{I_0}\right)S_{j,obs}^2
=O\!\left(\frac{1}{I}\right),
\]
since $I_1\asymp I$ and $I_0\asymp I$. Therefore,
\begin{equation}\label{eq:muhat-rate-perm}
\hat{\mu}_j^\Delta(\W\mid \mathbf{Y}^{obs},\mathcal{C}) = O_p(I^{-1/2})
\qquad\text{under }P_\pi.
\end{equation}

Now consider the sample variance across the $J_0$ control sellers,
\[
s^2_{\hat{\mu}^{\Delta}}(\W\mid \mathbf{Y}^{obs},\mathcal{C})
=
\frac{1}{J_0-1}\sum_{j:w^{S,obs}_j=0}
\left(\hat{\mu}_j^\Delta-\bar{\hat{\mu}}^\Delta\right)^2,
\qquad
\bar{\hat{\mu}}^\Delta:=\frac{1}{J_0}\sum_{j:w^{S,obs}_j=0}\hat{\mu}_j^\Delta.
\]
Using the deterministic inequality
\[
s^2_{\hat{\mu}^{\Delta}}
\le
\frac{J_0}{J_0-1}\cdot \frac{1}{J_0}\sum_{j:w^{S,obs}_j=0}\left(\hat{\mu}_j^\Delta\right)^2,
\]
taking $E_\pi[\cdot]$ on both sides and applying \eqref{eq:var-muhat-perm} yields
\begin{align*}
E_\pi\!\left[s^2_{\hat{\mu}^{\Delta}}(\W\mid \mathbf{Y}^{obs},\mathcal{C})\right]
&\le
\frac{J_0}{J_0-1}\cdot \frac{1}{J_0}
\sum_{j:w^{S,obs}_j=0}
E_\pi\!\left[\left(\hat{\mu}_j^\Delta\right)^2\right] \\
&=
\frac{J_0}{J_0-1}\cdot \frac{1}{J_0}
\sum_{j:w^{S,obs}_j=0}
\left(\frac{1}{I_1}+\frac{1}{I_0}\right)S_{j,obs}^2
=
O\!\left(\frac{1}{I}\right),
\end{align*}
where the last order uses $I_1\asymp I_0\asymp I$ and that $\frac{1}{J_0}\sum_{j:w^{S,obs}_j=0}S_{j,obs}^2$ is $O_p(1)$ under the sampling distribution (which follows from the moment assumptions). Since $s^2_{\hat{\mu}^{\Delta}}\ge 0$, Markov's inequality implies
\begin{equation}\label{eq:s2-muhat-rate-perm}
s^2_{\hat{\mu}^{\Delta}}(\W\mid \mathbf{Y}^{obs},\mathcal{C}) = O_p(I^{-1})
\qquad\text{under }P_\pi.
\end{equation}
Substituting \eqref{eq:s2-muhat-rate-perm} into the definition \eqref{eq:VS} gives
\[
V^{S}(\W\mid \mathbf{Y}^{obs},\mathcal{C})
=
\left(1-\frac{J_0}{J}\right)\frac{s^2_{\hat{\mu}^{\Delta}}(\W\mid \mathbf{Y}^{obs},\mathcal{C})}{J_0}
=
O_p\!\left(\frac{1}{I}\cdot\frac{1}{J_0}\right)
=
O_p\!\left(\frac{1}{IJ}\right)
\qquad\text{under }P_\pi,
\]
since $J_0\asymp J$.

\textbf{Step 1(b): Rate of $V^{B}$ under $P_\pi$.}
Conditional on $(\mathbf{Y}^{obs},\mathcal{C})$, the collection $\{\bar Y_i^B\}_{i=1}^I$ is fixed and the treated group is an SRSWOR subset of $[I]$ of size $I_1$. Let
\[
S_{\bar Y}^2:=\frac{1}{I-1}\sum_{i=1}^I\left(\bar Y_i^B-\bar Y_\cdot^B\right)^2,
\qquad
\bar Y_\cdot^B:=\frac{1}{I}\sum_{i=1}^I\bar Y_i^B.
\]
Under SRSWOR, the within-arm sample variances $s_1^2(\W\mid \mathbf{Y}^{obs},\mathcal{C})$ and $s_0^2(\W\mid \mathbf{Y}^{obs},\mathcal{C})$ are unbiased for $S_{\bar Y}^2$; in particular,
\[
E_\pi\!\left[s_1^2(\W\mid \mathbf{Y}^{obs},\mathcal{C})\right]=S_{\bar Y}^2,
\qquad
E_\pi\!\left[s_0^2(\W\mid \mathbf{Y}^{obs},\mathcal{C})\right]=S_{\bar Y}^2.
\]
Therefore,
\[
E_\pi\!\left[V^{B}(\W\mid \mathbf{Y}^{obs},\mathcal{C})\right]
=
\left(\frac{1}{I_1}+\frac{1}{I_0}\right)S_{\bar Y}^2
=
O\!\left(\frac{1}{I}\right),
\]
and since $V^{B}\ge 0$, Markov's inequality yields
\begin{equation}\label{eq:VB-rate-perm}
V^{B}(\W\mid \mathbf{Y}^{obs},\mathcal{C})=O_p(I^{-1})
\qquad\text{under }P_\pi.
\end{equation}
Moreover, on events where $S_{\bar Y}^2$ is bounded away from zero (which occur with probability approaching one under the sampling distribution by the nondegeneracy condition in Assumption \ref{assump:buyer-nondegeneracy}), the same SRSWOR LLN implies
\[
I\,V^{B}(\W\mid \mathbf{Y}^{obs},\mathcal{C})
=
\left(\frac{I}{I_1}+\frac{I}{I_0}\right)S_{\bar Y}^2 + o_p(1)
\asymp 1
\qquad\text{under }P_\pi,
\]
so $1/V^{B}=O_p(I)$ under $P_\pi$ on those events.

\textbf{Step 1(c): Conclude randomization equivalence.}
Combining \eqref{eq:s2-muhat-rate-perm} and \eqref{eq:VB-rate-perm} gives
\[
\frac{V^{S}(\W\mid \mathbf{Y}^{obs},\mathcal{C})}{V^{B}(\W\mid \mathbf{Y}^{obs},\mathcal{C})}
=
O_p(J^{-1})
\to 0
\qquad\text{under }P_\pi.
\]
Consequently,
\[
T^{WB,TW}(\W\mid \mathbf{Y}^{obs},\mathcal{C})
=
T^{WB}(\W\mid \mathbf{Y}^{obs},\mathcal{C})
\cdot
\left(1+\frac{V^{S}}{V^{B}}\right)^{-1/2}
=
T^{WB}(\W\mid \mathbf{Y}^{obs},\mathcal{C})\cdot(1+o_p(1)),
\]
under $P_\pi$, which implies that the conditional randomization distribution of $T^{WB,TW}$ also converges to $N(0,1)$.

\textbf{Step 2: Sampling distribution under $H_0^{wb,1}$.}
In this step we work under the \emph{sampling} probability induced by the two-sided complete randomization design. For notational brevity, we suppress the superscript ``obs'' and write
\[
(w^B,w^S)=\W^{obs},\qquad \mathbf{Y}^{obs}=\mathbf{Y}(\W^{obs}),
\]
and we abbreviate
\begin{align*}
    V^{B}(w^B,w^S)&:=V^{B}(\W^{obs}\mid \mathbf{Y}^{obs},\mathcal{C})\\
    V^{S}(w^B,w^S)&:=V^{S}(\W^{obs}\mid \mathbf{Y}^{obs},\mathcal{C}) \\
    s^2_{\hat{\mu}^{\Delta}}(w^B,w^S)&:=s^2_{\hat{\mu}^{\Delta}}(\W^{obs}\mid \mathbf{Y}^{obs},\mathcal{C})~.
\end{align*}
and similarly for $T^{B}$ and $T^{WB,TW}$.

Under $H_0^{wb,1}$ we have $\tau=0$ but $\tau(w^S)$ is random because it depends on the seller assignment. Write the decomposition
\[
T^{B}(w^B,w^S)
=
\tau(w^S)
+
\left\{T^{B}(w^B,w^S)-\tau(w^S)\right\}.
\]
By the two-stage randomization structure, $w^S$ is drawn first and then $w^B$ is drawn independently conditional on $w^S$. Under Assumptions~\ref{assump:bounded-fourth-moments}--\ref{assump:uniform-integrability}, the seller-side SRSWOR CLT yields
\[
\sqrt{J}\,\tau(w^S)\Rightarrow N\!\left(0,\frac{\pi_2}{1-\pi_2}S_{\mu^\Delta}^2\right),
\]
while the buyer-side CRD CLT (conditional on $w^S$) yields
\[
\sqrt{I}\,\left\{T^{B}(w^B,w^S)-\tau(w^S)\right\}\Rightarrow N\!\left(0,D-S_\delta^2\right),
\]
with asymptotic independence between the two terms. 

\textbf{Step 2(a): Control of $I\,V^{B}(w^B,w^S)$.}
By the LLN for the standard Neyman variance (as in the proof of the Neyman-style studentization result),
\[
I\,V^{B}(w^B,w^S)\xrightarrow{p} D,
\qquad
D:=\frac{S_1^2}{\pi_1}+\frac{S_0^2}{\pi_0}.
\]

\textbf{Step 2(b):Limit for $s^2_{\hat{\mu}^{\Delta}}(w^B,w^S)$ when $I$ and $J$ grow together.}
Recall $\mu_j^\Delta:=I^{-1}\sum_{i=1}^I\Delta_{ij}$ and that for each control seller $j\in S_0(w^S):=\{j:w_j^S=0\}$,
\[
\hat{\mu}_j^\Delta(w^B,w^S)
:=
\frac{1}{I_1}\sum_{i:w_i^B=1}Y_{ij}(1,0)
-
\frac{1}{I_0}\sum_{i:w_i^B=0}Y_{ij}(0,0),
\]
so that $\hat{\mu}_j^\Delta$ is the within-seller difference-in-means across buyers computed using $Y_{ij}(w_i^B,0)$ (because $w_j^S=0$ for $j\in S_0(w^S)$). Define the within-seller estimation error
\[
e_j:=\hat{\mu}_j^\Delta(w^B,w^S)-\mu_j^\Delta,
\qquad (j\in S_0(w^S)).
\]

\emph{(i) The buyer-randomization estimation error is negligible uniformly in $J$.}
Conditional on $w^S$, the assignment $w^B$ is a CRD over buyers with fixed margins $(I_1,I_0)$ and is independent of $w^S$. Standard finite-population randomization algebra implies
\[
E\!\left[e_j\mid w^S\right]=0
\quad\text{and}\quad
E\!\left[e_j^2\mid w^S\right]
=
\Var\!\left(\hat{\mu}_j^\Delta(w^B,w^S)\mid w^S\right)
=
\frac{S_{j,1}^2}{I_1}+\frac{S_{j,0}^2}{I_0}-\frac{S_{j,\Delta}^2}{I},
\]
where, for $z\in\{0,1\}$,
\[
S_{j,z}^2:=\frac{1}{I-1}\sum_{i=1}^I\left(Y_{ij}(z,0)-\bar Y_{\cdot j}(z,0)\right)^2,
\qquad
S_{j,\Delta}^2:=\frac{1}{I-1}\sum_{i=1}^I\left(\Delta_{ij}-\mu_j^\Delta\right)^2,
\]
and $\bar Y_{\cdot j}(z,0)=I^{-1}\sum_i Y_{ij}(z,0)$.

Since $I_1/I\to\pi_1\in(0,1)$ and $I_0/I\to\pi_0\in(0,1)$, there exists a constant $C<\infty$ such that for all large $N$,
\[
E\!\left[e_j^2\mid w^S\right]
\le
\frac{C}{I}\left(S_{j,1}^2+S_{j,0}^2+S_{j,\Delta}^2\right).
\]
By Assumption~\ref{assump:bounded-fourth-moments} and Jensen's inequality, the averages
\[
\frac{1}{J}\sum_{j=1}^J S_{j,1}^2,\quad \frac{1}{J}\sum_{j=1}^J S_{j,0}^2,\quad \frac{1}{J}\sum_{j=1}^J S_{j,\Delta}^2
\]
are $O(1)$ uniformly in $N$.\footnote{One convenient bound is
$\frac{1}{J}\sum_{j=1}^J S_{j,z}^2 \le \frac{2}{IJ}\sum_{i,j} Y_{ij}(z,0)^2$
and similarly for $S_{j,\Delta}^2$, and bounded fourth moments imply bounded second moments by Cauchy--Schwarz.}
Because $S_0(w^S)$ is an SRSWOR subset of $[J]$ with $J_0\asymp J$, the corresponding subset averages are also $O_p(1)$. Therefore,
\[
E\!\left[\frac{1}{J_0}\sum_{j:w_j^S=0} e_j^2 \,\bigg|\, w^S\right]
\le
\frac{C}{I}\cdot \frac{1}{J_0}\sum_{j:w_j^S=0}\left(S_{j,1}^2+S_{j,0}^2+S_{j,\Delta}^2\right)
=
O_p\!\left(\frac{1}{I}\right).
\]
Since the left-hand side is nonnegative, a conditional Markov inequality yields
\begin{equation}\label{eq:avg-ej2-vanish}
\frac{1}{J_0}\sum_{j:w_j^S=0} e_j^2 \xrightarrow{p} 0,
\end{equation}
and hence also $s_e^2:=\frac{1}{J_0-1}\sum_{j:w_j^S=0}(e_j-\bar e)^2\xrightarrow{p}0$ with $\bar e=J_0^{-1}\sum_{j:w_j^S=0}e_j$.

\emph{(ii) The cross-seller sample variance of $\mu_j^\Delta$ over control sellers is consistent.}
Let
\[
\bar\mu_0(w^S):=\frac{1}{J_0}\sum_{j:w_j^S=0}\mu_j^\Delta,
\qquad
\overline{\mu^2}_0(w^S):=\frac{1}{J_0}\sum_{j:w_j^S=0}(\mu_j^\Delta)^2,
\]
and denote the full-population counterparts
\[
\bar\mu:=\frac{1}{J}\sum_{j=1}^J\mu_j^\Delta=\tau,
\qquad
\overline{\mu^2}:=\frac{1}{J}\sum_{j=1}^J(\mu_j^\Delta)^2.
\]
Since $S_0(w^S)$ is SRSWOR, the SRSWOR variance formula gives
\[
\Var\!\big(\bar\mu_0(w^S)\big)=\left(1-\frac{J_0}{J}\right)\frac{S_{\mu^\Delta}^2}{J_0}=O\!\left(\frac{1}{J}\right),
\]
so $\bar\mu_0(w^S)-\bar\mu=o_p(1)$. Under $H_0^{wb,1}$, $\bar\mu=\tau=0$, hence $\bar\mu_0(w^S)=o_p(1)$.

Next, apply the same SRSWOR variance formula to the finite population $\{(\mu_j^\Delta)^2\}_{j=1}^J$:
\[
\Var\!\big(\overline{\mu^2}_0(w^S)\big)
=
\left(1-\frac{J_0}{J}\right)\frac{S_{\mu^2}^2}{J_0},
\qquad
S_{\mu^2}^2:=\frac{1}{J-1}\sum_{j=1}^J\left((\mu_j^\Delta)^2-\overline{\mu^2}\right)^2.
\]
By Jensen and Assumption~\ref{assump:bounded-fourth-moments},
\[
\frac{1}{J}\sum_{j=1}^J(\mu_j^\Delta)^4
=
\frac{1}{J}\sum_{j=1}^J\left(\frac{1}{I}\sum_{i=1}^I\Delta_{ij}\right)^4
\le
\frac{1}{IJ}\sum_{i=1}^I\sum_{j=1}^J\Delta_{ij}^4
=O(1),
\]
so $S_{\mu^2}^2=O(1)$ and therefore $\Var(\overline{\mu^2}_0(w^S))=O(J^{-1})$. Chebyshev's inequality implies $\overline{\mu^2}_0(w^S)-\overline{\mu^2}=o_p(1)$.

Now write the sample variance of $\{\mu_j^\Delta\}_{j:w_j^S=0}$ as
\[
s_{\mu^\Delta}^2(w^S)
=
\frac{1}{J_0-1}\sum_{j:w_j^S=0}\left(\mu_j^\Delta-\bar\mu_0(w^S)\right)^2
=
\frac{J_0}{J_0-1}\left(\overline{\mu^2}_0(w^S)-\bar\mu_0(w^S)^2\right).
\]
Since $\overline{\mu^2}_0(w^S)\to_p \overline{\mu^2}$ and $\bar\mu_0(w^S)=o_p(1)$, and $J_0/(J_0-1)\to 1$, we obtain
\begin{equation}\label{eq:s-muDelta-consistent}
s_{\mu^\Delta}^2(w^S)\xrightarrow{p} \overline{\mu^2}.
\end{equation}
Under $H_0^{wb,1}$, $\tau=0$, so
\[
S_{\mu^\Delta}^2=\frac{1}{J-1}\sum_{j=1}^J(\mu_j^\Delta-\tau)^2=\frac{1}{J-1}\sum_{j=1}^J(\mu_j^\Delta)^2
=
\frac{J}{J-1}\,\overline{\mu^2},
\]
and since $J/(J-1)\to 1$, \eqref{eq:s-muDelta-consistent} implies $s_{\mu^\Delta}^2(w^S)\to_p S_{\mu^\Delta}^2$.

\emph{(iii) Transfer from $\mu_j^\Delta$ to $\hat{\mu}_j^\Delta$.}
Using $\hat{\mu}_j^\Delta=\mu_j^\Delta+e_j$ and the decomposition of sample variance,
\[
s^2_{\hat{\mu}^{\Delta}}(w^B,w^S)
=
s_{\mu^\Delta}^2(w^S)
+
s_e^2(w^B,w^S)
+
\frac{2}{J_0-1}\sum_{j:w_j^S=0}\left(\mu_j^\Delta-\bar\mu_0(w^S)\right)\left(e_j-\bar e\right),
\]
the Cauchy--Schwarz inequality yields the bound
\[
\left|
\frac{2}{J_0-1}\sum_{j:w_j^S=0}\left(\mu_j^\Delta-\bar\mu_0(w^S)\right)\left(e_j-\bar e\right)
\right|
\le
2\sqrt{s_{\mu^\Delta}^2(w^S)}\sqrt{s_e^2(w^B,w^S)}.
\]
By \eqref{eq:avg-ej2-vanish} we have $s_e^2(w^B,w^S)\to_p 0$, and by the previous paragraph $s_{\mu^\Delta}^2(w^S)=O_p(1)$. Therefore the cross term is $o_p(1)$ and
\[
s^2_{\hat{\mu}^{\Delta}}(w^B,w^S)
=
s_{\mu^\Delta}^2(w^S)+o_p(1)
\xrightarrow{p}
S_{\mu^\Delta}^2.
\]
This establishes the desired consistency of $s^2_{\hat{\mu}^{\Delta}}(w^B,w^S)$ under joint asymptotics with $I\to\infty$, $J\to\infty$, and $I/J\to\lambda\in(0,\infty)$.

\textbf{Step 2(c): Limit of $I\,V^{S}(w^B,w^S)$ and $I\,V^{TW}(w^B,w^S)$.}
Substituting $s^2_{\hat{\mu}^{\Delta}}(w^B,w^S)\to_p S_{\mu^\Delta}^2$ into \eqref{eq:VS} gives
\[
I\,V^{S}(w^B,w^S)
=
I\left(1-\frac{J_0}{J}\right)\frac{s^2_{\hat{\mu}^{\Delta}}(w^B,w^S)}{J_0}
=
\frac{I}{J}\cdot \frac{J_1}{J_0}\cdot s^2_{\hat{\mu}^{\Delta}}(w^B,w^S)
\xrightarrow{p}
\lambda\cdot\frac{\pi_2}{1-\pi_2}\,S_{\mu^\Delta}^2.
\]
Combining with $I\,V^{B}(w^B,w^S)\to_p D$ yields
\[
I\,V^{TW}(w^B,w^S)
=
I\,V^{B}(w^B,w^S)+I\,V^{S}(w^B,w^S)
\xrightarrow{p}
D+\lambda\cdot\frac{\pi_2}{1-\pi_2}\,S_{\mu^\Delta}^2
=:D_{TW}.
\]

\textbf{Step 2(d): Conclude the sampling limit of $T^{WB,TW}$.}
Since
\[
\sqrt{I}\,T^{B}(w^B,w^S)
=
\sqrt{I}\,\tau(w^S)
+
\sqrt{I}\,\{T^{B}(w^B,w^S)-\tau(w^S)\},
\]
and $\sqrt{I}\,\tau(w^S)=\sqrt{I/J}\cdot \sqrt{J}\tau(w^S)\Rightarrow N\!\left(0,\lambda\cdot\frac{\pi_2}{1-\pi_2}S_{\mu^\Delta}^2\right)$, asymptotic independence yields
\[
\sqrt{I}\,T^{B}(w^B,w^S)
\Rightarrow
N\!\left(
0,\ (D-S_\delta^2)+\lambda\cdot\frac{\pi_2}{1-\pi_2}S_{\mu^\Delta}^2
\right).
\]
By Slutsky's theorem,
\[
T^{WB,TW}(w^B,w^S)
=
\frac{\sqrt{I}\,T^{B}(w^B,w^S)}{\sqrt{I\,V^{TW}(w^B,w^S)}}
\Rightarrow
N\!\left(
0,\
\frac{(D-S_\delta^2)+\lambda\frac{\pi_2}{1-\pi_2}S_{\mu^\Delta}^2}{D_{TW}}
\right)
=
N\!\left(0,\ 1-\frac{S_\delta^2}{D_{TW}}\right),
\]
which has variance weakly smaller than $1$.

\textbf{Step 3: Asymptotic validity.}
By Step 1, conditional on $(\mathbf{Y}^{obs},\mathcal{C})$, the randomization distribution of $T^{WB,TW}(\W\mid \mathbf{Y}^{obs},\mathcal{C})$ converges to $N(0,1)$. Hence the randomization $p$-value satisfies
\[
p_N
=
\left\{1-\Phi\!\left(T^{WB,TW}(\W^{obs}\mid \mathbf{Y}^{obs},\mathcal{C})\right)\right\}
+o_p(1).
\]
By Step 2, under $H_0^{wb,1}$,
\[
T^{WB,TW}(\W^{obs}\mid \mathbf{Y}^{obs},\mathcal{C})
\Rightarrow
\sigma_{TW}Z,
\qquad
Z\sim N(0,1),
\qquad
\sigma_{TW}^2=1-\frac{S_\delta^2}{D_{TW}}\le 1.
\]
Let $z_{1-\alpha}$ denote the $(1-\alpha)$-quantile of $N(0,1)$. Then
\begin{align*}
\limsup_{I,J\to\infty}P(p_N\le \alpha)
&=
\limsup_{I,J\to\infty}P\left(1-\Phi(T^{WB,TW}(\W^{obs}))\le \alpha\right)\\
&=
\limsup_{I,J\to\infty}P\left(T^{WB,TW}(\W^{obs})\ge z_{1-\alpha}\right)\\
&\le
P(\sigma_{TW}Z\ge z_{1-\alpha})
=
1-\Phi\!\left(\frac{z_{1-\alpha}}{\sigma_{TW}}\right)
\le
1-\Phi(z_{1-\alpha})
=
\alpha,
\end{align*}
where the last inequality uses $\sigma_{TW}\le 1$, which implies $z_{1-\alpha}/\sigma_{TW}\ge z_{1-\alpha}$. This establishes asymptotic validity under $H_0^{wb,1}$.
\end{proof}

\newpage
\section{Additional Simulation Results}\label{app:additional-simulation}
\subsection{Full simulation study for the weak null} \label{app:additional-simulation-global-weak}
In this section, we present the full simulation study conducted for weak null hypotheses as in Section \ref{subsec:sim-weak-null}. We consider three DGPs that all satisfy the $H_0^{wb,1}$ (population average) but differ in whether $H_0^{wb,2}$ (seller-specific) holds and whether Assumption~\ref{assump:buyer-nondegeneracy} is satisfied.

Let $\tilde{Y}_{ij}(0,0)$ and $\tilde{\Delta}_{ij}$ denote the baseline i.i.d.\ components generated as in Section \ref{subsec:sim-sharp}:
\[
\tilde{Y}_{ij}(0,0)\stackrel{\text{ind}}{\sim}\mathcal{N}(0,0.2^2),
\qquad
\tilde{\Delta}_{ij}\stackrel{\text{ind}}{\sim}\mathcal{N}(0,0.4^2),
\]
and define $Y_{ij}(1,0)=Y_{ij}(0,0)+\Delta_{ij}$ with $\Delta_{ij}$ specified below. (The remaining potential outcomes $Y_{ij}(0,1)$ and $Y_{ij}(1,1)$ are generated as in the Section \ref{subsec:sim-sharp} and play no role in the buyer-spillover test considered here.)

\begin{enumerate}
    \item \textbf{Model 1 (dyadic i.i.d.\ outcomes and treatment effects).}
    Set
    \[
        Y_{ij}(0,0)=\tilde{Y}_{ij}(0,0),
        \qquad
        \Delta_{ij}=\tilde{\Delta}_{ij}.
    \]
    Because treatment effects are generated i.i.d.\ across dyads with mean zero, $H_0^{wb,2}$ is approximately satisfied for large focal sets (local averages concentrate around zero), and $H_0^{wb,1}$ holds by construction.

    \item \textbf{Model 2 (buyer fixed effects in outcomes and seller fixed effects in treatment effects).}
    Draw independent fixed effects
    \[
        \alpha_i \stackrel{\text{ind}}{\sim} \mathcal{N}(0,0.1^2),
        \qquad
        \beta_j \stackrel{\text{ind}}{\sim} \mathcal{N}(0,0.4^2),
    \]
    and set
    \[
        Y_{ij}(0,0)=\tilde{Y}_{ij}(0,0)+\alpha_i,
        \qquad
        \Delta_{ij}=\tilde{\Delta}_{ij}+\beta_j.
    \]
    This DGP violates $H_0^{wb,2}$ because seller-specific average treatment effects vary across $j$ through $\beta_j$, so $\tau(w^S)$ need not be close to zero for a realized seller focal set. However, the $H_0^{wb,1}$ remains satisfied because both $\tilde{\Delta}_{ij}$ and $\beta_j$ have mean zero. Moreover, the buyer fixed effects $\alpha_i$ ensure Assumption~\ref{assump:buyer-nondegeneracy} is satisfied by inducing persistent cross-buyer heterogeneity in buyer-level averages.

    \item \textbf{Model 3 (seller fixed effects in treatment effects only).}
    Draw $\beta_j \stackrel{\text{ind}}{\sim} \mathcal{N}(0,0.4^2)$ and set
    \[
        Y_{ij}(0,0)=\tilde{Y}_{ij}(0,0),
        \qquad
        \Delta_{ij}=\tilde{\Delta}_{ij}+\beta_j.
    \]
    As in Model 2, $H_0^{wb,2}$ is violated while $H_0^{wb,1}$ is satisfied. However, without buyer fixed effects, buyer-level averages concentrate as $J$ grows, so Assumption~\ref{assump:buyer-nondegeneracy} is violated in this design.
\end{enumerate}

\begin{table}[ht!]
\centering
\setlength{\tabcolsep}{10pt}
\begin{adjustbox}{max width=\linewidth,center}
\begin{tabular}{ccccc}
\toprule
 &      &  \multicolumn{3}{c}{Under $H_0^{wb,1}$}  \\ \cmidrule{3-5} 
Setup & $n$ & {FRT} & {FRT adjusted} & {FRT two-way}   \\ \midrule
&10  & 11.24 & 4.52 & 4.38 \\
Model 1: i.i.d.\ outcomes
&20  & 9.70  & 3.98 & 4.08 \\
and treatment effects
&30  & 8.88  & 3.40 & 3.40 \\
(focal weak null holds)
&40  & 8.46  & 3.74 & 3.28 \\
&50  & 9.20  & 3.56 & 3.42 \\
&100 & 8.60  & 3.62 & 3.44 \\
\\
&10  & 11.04 & 9.52 & 4.04 \\ 
Model 2: fixed effects in 
&20  & 16.34 & 15.16 & 4.64 \\
both outcomes and
&30  & 16.20 & 15.10 & 4.76 \\
treatment effects
&40  & 21.62 & 20.24 & 4.90 \\
(Assumption~\ref{assump:buyer-nondegeneracy} satisfied)
&50  & 19.52 & 17.54 & 5.04 \\
&100 & 22.18 & 21.68 & 5.46 \\
\\
&10  & 38.00 & 24.36 & 6.42 \\ 
Model 3: i.i.d.\ outcomes 
&20  & 52.64 & 41.88 & 8.82 \\
and fixed effects 
&30  & 58.82 & 50.26 & 8.62 \\
in treatment effects
&40  & 57.52 & 47.12 & 7.92 \\
(Assumption~\ref{assump:buyer-nondegeneracy} violated)
&50  & 67.64 & 59.70 & 9.60 \\
&100 & 75.38 & 70.02 & 8.98 \\
\bottomrule
\end{tabular}
\end{adjustbox}
\caption{Rejection probabilities (in percent) for testing the global weak null $H_0^{wb,1}$, based on 5{,}000 Monte Carlo replications and 500 permutations per replication.}
\label{table:size-global-weak-app}
\end{table}

Table~\ref{table:size-global-weak-app} reports the resulting rejection probabilities under $H_0^{wb,1}$. Under Model~1, both FRT adjusted and FRT two-way are close to the nominal 5\% level (with mild conservativeness for larger $n$). Under Model~2, only the two-way studentized procedure controls size, whereas the standard studentization fails due to focal-average randomness induced by seller heterogeneity. Under Model~3, even the two-way procedure fails to control size, which is consistent with our theory because Assumption~\ref{assump:buyer-nondegeneracy} is violated.

\subsection{Remarks and simulation study for weak null of total effects}\label{app:total-effect-weak}

Following the analysis in Section \ref{subsec:weak-null-standard}, the asymptotic properties of tests for the total effect mirror those established for the buyer-spillover effect when using Procedure \ref{procedure:weak-spillover}. Under a weak null hypothesis concerning the global average total effect, the test is generally expected to be asymptotically invalid. Conversely, if the weak null is defined relative to the focal average total effect, the test maintains asymptotic validity. These theoretical insights are consistent with the simulation results presented in this section.

Under the same experimental framework of the weak null $H_0^{wb,2}$ in Section \ref{subsec:sim-weak-null}, Table~\ref{table:size-power-weak} presents the rejection probabilities under the weak null and alternative hypotheses. In this setup, the block size is fixed at $2$, ensuring a sufficiently large number of blocks to guarantee the precision of the variance estimators. In this setup, the weak null hypotheses of zero focal average effects are satisfied by the data generating process (DGP). 

\begin{table}[ht!]
\centering
\setlength{\tabcolsep}{3pt}
\begin{adjustbox}{max width=0.9\linewidth,center}
\begin{tabular}{ccccccccccc}
\toprule
 &      &  \multicolumn{3}{c}{Under $H_0$} & & \multicolumn{3}{c}{Under $H_1$} \\ \cmidrule{3-5} \cmidrule{7-9}
Setup & $n$ & {FRT} & {Neymanian} & {FRT adjusted}  &  &  {FRT} & {Neymanian} & {FRT adjusted}    \\ \midrule
&10 & 13.84  & 0.82  & 8.22   &  & 15.24  & 1.36  & 8.16  \\ 
I.I.D outcomes
&20 & 12.14  & 1.18  & 6.04   &  & 16.48  & 5.36  & 10.64  \\
and treatment effects
&30 & 11.90  & 1.30  & 5.70   &  & 19.78  & 10.48  & 13.36  \\
(similar to $H_0^{wb,2}$)
&40 & 11.46  & 0.84  & 5.16   &  & 24.76  & 22.16  & 16.74  \\
&50 & 11.56  & 1.18  & 5.52   &  & 30.74  & 38.04  & 21.30  \\
&100 & 11.44  & 0.86  & 4.84  &  & 67.32  & 97.82  & 53.56  \\
\bottomrule
\end{tabular}
\end{adjustbox}
\caption{Rejection probabilities under weak the null of zero focal average effect and alternative hypothesis}
\label{table:size-power-weak}
\end{table}

Following the analysis in Section \ref{subsec:weak-null-twoway}, the two-way studentization exploits the fact that, under $H_0^{\mathrm{wb},1}$, the additional sampling uncertainty induced by seller randomization can be consistently estimated from the data via cross-seller variation in estimable seller-level effects (cf.\ \eqref{eq:mu-hat-delta}--\eqref{eq:VS}). In contrast, extending this correction to global total effects would require estimating the variation of a focal total effect across the randomization structure, which is generally infeasible under weak nulls in this settings because the relevant total-effect variation is not directly estimable. For instance, with block size $K=1$, the focal average need not converge to the global average even asymptotically, so discrepancies between focal and global total effects can persist. Consequently, a variance correction of the form \eqref{eq:VS} is generally insufficient, and valid studentization for global total-effect weak nulls is substantially more difficult.

\newpage
\section{Additional Details and Results for Data Application}\label{app:details-data}
In this section, we provide extra details and results for the data application in Section \ref{sec:empirical}, including data dictionary, mapping between notation and empirical setting and testing results on additional outcome variables. 

\paragraph{Data dictionary}
Figure \ref{fig:variables-data} presents a list of variables available in the data set of \cite{Comola2021}, including several household-pair covariates measured at both survey waves. These include the absolute difference between paired households in the number of household members, number of children, whether the female household head is married, and indicators for whether either household experienced a livestock or death shock.

\begin{figure}[h]
    \centering
    \includegraphics[width=1.0\linewidth]{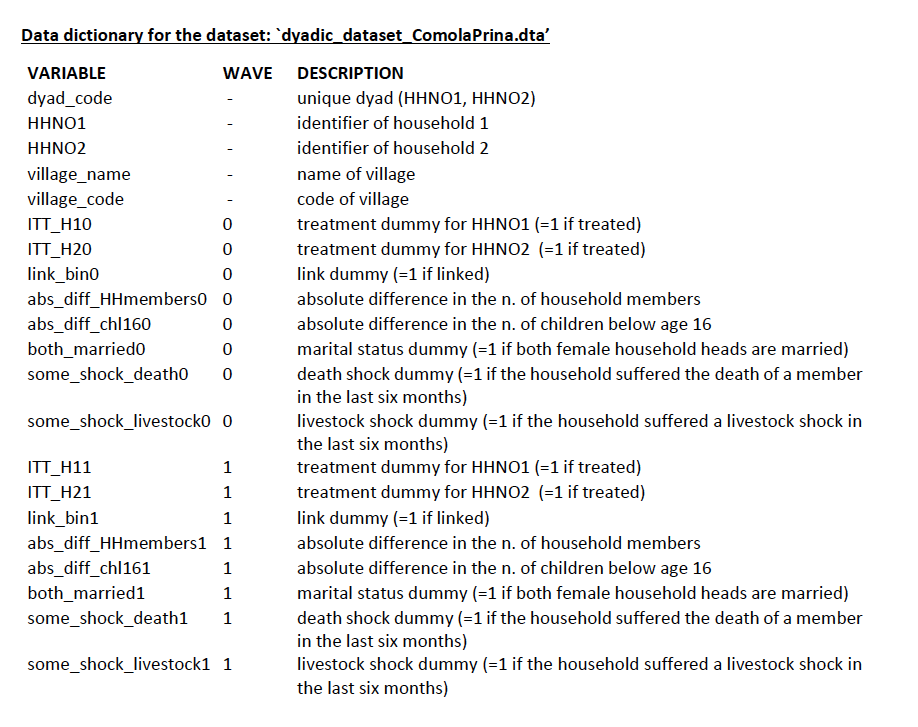}
    \caption{All available variables from the data}
    \label{fig:variables-data}
\end{figure}

\paragraph{Mapping between notation and empirical setting}
In table \ref{table:structure}, we summarize how our notation  is mapped to this empirical data application.

\begin{table}[ht!]
\centering
\vspace{20px}
\setlength{\tabcolsep}{8pt}
\begin{adjustbox}{max width=\linewidth,center}
\begin{tabular}{ll}
\toprule
\textbf{Notation} & \textbf{Description} \\
\midrule
$(i, j)$ & Household pair where $i$ experienced a shock (buyer) and $j$ did not (seller) \\
$X_i$ & Covariate indicating whether household $i$ experienced a shock ($X_i = 1$) \\
$Y_{i,j}$ & Observed change in financial link between households $i$ and $j$ \\
         & (alternative outcomes are also considered in Table~\ref{table:empirical}) \\
$w^B$ & Vector indicating which households (with shocks) were offered savings accounts \\
$w^S$ & Vector indicating which households (without shocks) were offered savings accounts \\
$Y_{i,j}(0, 0)$ & Benchmark outcome when neither household is treated \\
$Y_{i,j}(1, 0)$ & Financial link change when only the shocked household $i$ is treated \\
$Y_{i,j}(0, 1)$ & Financial link change when only the non-shocked household $j$ is treated \\
$Y_{i,j}(1, 1)$ & Financial link change when both households are treated \\
\bottomrule
\end{tabular}
\end{adjustbox}
\caption{Mapping between empirical setting and notation used in the paper}
\label{table:structure}
\end{table}

\paragraph{Additional outcome variables}
In addition to the main outcome of interest from \citet{Comola2021}—changes in network links—we also constructed auxiliary outcomes based on changes in household-pair covariates measured at baseline and endline. For instance, the dataset includes information about the absolute difference in household size between each pair in both survey waves; differencing these values across time yields outcomes that reflect changes in household characteristics such as size. 
However, as shown in Table \ref{table:empirical-app},  we don't find significant effects on these outcomes as well.

\begin{table}[ht!]
\centering
\vspace{20px}
\setlength{\tabcolsep}{8pt}
\begin{adjustbox}{max width=\linewidth,center}
\begin{tabular}{ccccc}
\toprule
Outcome & Hypothesis & FRT & FRT adjusted & Two-sample $t$-test    \\ \midrule
\multirow{3}{*}{Financial network links} & $Y_{i,j}(0,0) = Y_{i,j}(1,0)$ & 0.49 & 0.50 & 0.49\\ 
 & $Y_{i,j}(0,0) = Y_{i,j}(0,1)$ & 0.42 & 0.41 & 0.41\\ 
 & $Y_{i,j}(0,0) = Y_{i,j}(1,1)$ & 0.70 & 0.69 & 0.77  \\ \\
 \multirow{3}{*}{Number of household members} & $Y_{i,j}(0,0) = Y_{i,j}(1,0)$ & 0.64 & 0.64 & 0.63\\ 
  & $Y_{i,j}(0,0) = Y_{i,j}(0,1)$ & 0.50 & 0.49 & 0.48\\ 
 & $Y_{i,j}(0,0) = Y_{i,j}(1,1)$ & 0.98 & 0.98 & 0.94  \\ \\
 \multirow{3}{*}{Number of children} & $Y_{i,j}(0,0) = Y_{i,j}(1,0)$ & 0.22 & 0.23 & 0.23\\ 
  & $Y_{i,j}(0,0) = Y_{i,j}(0,1)$ & 0.76 & 0.74 & 0.75\\ 
 & $Y_{i,j}(0,0) = Y_{i,j}(1,1)$ & 0.95 & 0.95 & 0.81  \\ 
\bottomrule
\end{tabular}
\end{adjustbox}
\caption{Randomization $p$-values on multiple outcomes defined 
in the data example of Section~\ref{sec:empirical}.}
\label{table:empirical-app}
\end{table}

The motivation for including these auxiliary outcomes arises from the observation that only a small fraction of household pairs are linked in the network, resulting in a sparse outcome matrix with limited variation, as shown in Table~\ref{table:descriptive}. To address this limitation, we include two additional outcomes in Table~\ref{table:empirical}—the absolute difference in the number of household members and in the number of children—which exhibit greater variation. The treatment effects remain insignificant for these outcomes, providing further support for our findings.

\begin{table}[ht!]
\centering
\vspace{20px}
\setlength{\tabcolsep}{8pt}
\begin{adjustbox}{max width=\linewidth,center}
\begin{tabular}{ccc}
\toprule
Subsamples & Number of Non-zero Entries  & Sample Size  
\\ \midrule
$\{Y_i:w^B=0,w^S=0\}$ & 14 & 749 \\ 
$\{Y_i:w^B=1,w^S=0\}$ & 22 & 1049 \\ 
$\{Y_i:w^B=0,w^S=1\}$ & 9 & 697 \\ 
$\{Y_i:w^B=1,w^S=1\}$ & 21 & 1017 \\ 
\bottomrule
\end{tabular}
\end{adjustbox}
\caption{Summary statistics on network link}
\label{table:descriptive}
\end{table}

\end{document}